\documentclass[11pt]{amsart}
\usepackage[margin=1in]{geometry}
\usepackage{mathrsfs}
\usepackage{color}
\usepackage{graphicx}
\usepackage{mathtools}
\usepackage{bbm}
\usepackage[colorlinks=false, pdfstartview=FitV, linkcolor=blue, citecolor=blue, urlcolor=blue]{hyperref}

\DeclarePairedDelimiter\floor{\lfloor}{\rfloor}
\newcommand{\beq}{\begin{equation}}
\newcommand{\eeq}{\end{equation}}
 \newcommand\ben{\begin{equation*}}
 \newcommand\ebn{\end{equation*}}
\def\bU{{\bf U}}
\def\ep{\epsilon}
\def\eq{\begin{equation}}
\def\endeq{\end{equation}}
\def\bbm{\begin{bmatrix}}
\def\ebm{\end{bmatrix}}
\def\bpm{\begin{pmatrix}}
\def\epm{\end{pmatrix}}
\def\bvm{\begin{vmatrix}}
\def\evm{\end{vmatrix}}

\def\sg{\sigma}
\def\a{\alpha}
\def\ts{\widetilde{\sigma}}
\newcommand\lb{\left(}
\newcommand\rb{\right)}

\newtheorem{theorem}{Theorem}
\newtheorem{conjecture}{Conjecture}

\newtheorem{lemma}{Lemma}

\theoremstyle{definition}

\newtheorem{remark}{Remark}

\makeatletter
\@addtoreset{equation}{section}
\makeatother

\title[Characteristic polynomials and Painlev\'e equations]{A representation of joint moments of CUE characteristic polynomials in terms of Painlev\'e functions}
\author[E. Basor]{Estelle Basor}
\address[E. Basor]{American Institute of Mathematics, 600 E. Brokaw Rd, San Jose, CA 95112, USA}
\email{ebasor@aimath.org}
\author[P. Bleher]{Pavel Bleher}
\address[P. Bleher]{Department of Mathematical Sciences,
 Indiana University-Purdue University Indianapolis,
 402 N. Blackford St.,
 Indianapolis, IN 46202-3267, USA}
\email{pbleher@iupui.edu}
\author[R. Buckingham]{Robert Buckingham}
\address[R. Buckingham]{Department of Mathematical Sciences\\ University of Cincinnati\\ PO Box 210025\\ Cincinnati, OH 45221, USA.}
\email{buckinrt@uc.edu}
\author[T. Grava]{Tamara Grava}
\address[T. Grava]{Area of Mathematics, SISSA, via Bonomea 265 - 34136, Trieste, Italy, School of Mathematics, University of Bristol, Bristol, BS8 1TW, United Kingdom, and Aix-Marseille Universit\'e, CNRS, Soci\'et\'e Math\'ematique de France, CIRM (Centre International de Rencontres Math\'ematiques), Marseille, France}
\email{grava@sissa.it}
\author[A. Its]{Alexander Its}
\address[A. Its]{Department of Mathematical Sciences,
 Indiana University-Purdue University Indianapolis,
 402 N. Blackford St.,
 Indianapolis, IN 46202-3267, USA}
\email{aits@iupui.edu}
\author[E. Its]{Elizabeth Its}
\address[E. Its]{Department of Mathematical Sciences,
 Indiana University-Purdue University Indianapolis,
 402 N. Blackford St.,
 Indianapolis, IN 46202-3267, USA}
\email{enits@iupui.edu}
\author[J.P. Keating]{Jonathan P. Keating}
\address[J.P. Keating]{School of Mathematics, University of Bristol, Bristol, BS8 1TW, United Kingdom}
\email{j.p.keating@bristol.ac.uk}

\begin{document}
\begin{abstract}
We establish a representation of the joint moments of the characteristic polynomial of a CUE random 
matrix and its derivative in terms of a solution of the $\sigma$-Painlev\'e V equation.  The derivation involves the analysis of a formula for the joint moments in terms of a determinant of generalised Laguerre polynomials using the Riemann-Hilbert method.  We use this connection with the  $\sigma$-Painlev\'e V equation to derive explicit formulae for the joint moments and to show that in the large-matrix limit the joint moments are related to a solution of the $\sigma$-Painlev\'e III$^\prime$ equation.  Using the conformal block expansion of the $\tau$-functions associated with the $\sigma$-Painlev\'e V  and the $\sigma$-Painlev\'e III$^\prime$ equations leads to general conjectures for the joint moments.
\end{abstract}
\maketitle

\section{Introduction}
Let $U\in \bU(N)$ be taken from the Circular Unitary Ensemble (CUE) of random matrices. Consider its characteristic polynomial,
\begin{equation}\label{charact}
Z_{U}(\theta) := \prod_{n=1}^{N}\Bigl(1 - e^{i(\theta_n - \theta)}\Bigr),
\end{equation}
where $e^{i\theta_1}, ..., e^{i\theta_{N}}$ are the eigenvalues of $U$  and $\theta_i\in[0,2\pi)$. Put
\begin{equation}\label{Vdef}
V_{U}(\theta) := \exp\left(iN\frac{\theta+\pi}{2} - i\sum_{n=1}^{N}\frac{\theta_n}{2}\right)Z_{U}(\theta),
\end{equation}
so that $V_{U}(\theta) $ is real-valued for $\theta \in [0, 2\pi)$. The objects of our study  are the joint moments 
of the function $V_{U}(\theta) $ and its derivative,
\begin{equation}\label{FNdef}
F_{N}(h,k) := \int_{\bU(N)}|V_{U}(0)|^{2k-2h}|V'_{U}(0)|^{2h}d\mu^{\mbox{Haar}},
\end{equation}
where it is assumed that
$$
h > -\frac{1}{2}\quad \mbox{and}\quad k > h - \frac{1}{2}.
$$
Here $d\mu^{\mbox{Haar}}$ is Haar measure, the unique probability 
measure over the $N\times N$ unitary matrices that is invariant under the 
action of the unitary group \cite{Conway:1990}.

These joint moments have been the focus of a number of previous studies: when $h=0$ they can be computed in several different ways -- see, for example \cite{KeatingS:2000, BumpG}; similarly, when $h=k$ they can be computed using standard techniques \cite{ConreyRS:2006}; the general mixed moments for $h, k \in {\mathbb N}$ have been analysed by combining these approaches \cite{Hughes1, Hughes2, Dehaye:2008, Dehaye:2010, Winn:2012, Riedtmann}.  The results obtained suggest that in general $F_{N}(h,k)$ grows like $N^{k^2+2h}$ as $N\rightarrow\infty$ and it is a key problem to  prove this and then to evaluate the limit
\begin{equation}\label{Fdef}
F(h, k) := \lim_{N\rightarrow \infty}\frac{1}{N^{k^2 + 2h}} F_{N}(h,k).
\end{equation}

For $h, k \in {\mathbb N}$, $k > h - 1/2$, an expression for $F_{N}(h, k)$ 
was obtained in \cite{Dehaye:2008}  in terms of certain  sums over partitions.
A similar answer was also found in the case 
$h = (2m -1)/2$, $m \in {\mathbb N}$, $k \in {\mathbb N}$, $k > h - 1/2$  in 
\cite{Winn:2012} (see \cite{Winn:2012} also for a survey on other related 
results). However, these formulae do not allow for easy computation beyond the first few values of $k$ and $h$, in part because they are not recursive; also, they do not lead straightforwardly to formulae that extend to non-integer or non-half-integer values of $k$ and $h$.  The goal remains to analyse the large-$N$ limit of $F_{N}(h, k)$
and evaluate the quantity $F(h, k)$ for arbitrary real $k$ and $h$.

One motivation for studying the joint moments is as follows.  In 1973, Montgomery \cite{Montgomery:1973} conjectured that, assuming the 
Riemann hypothesis, the distances between appropriately normalised pairs of 
zeros of the Riemann zeta function follow a certain distribution previously 
shown by Dyson \cite{Dyson:1962} to describe spacings between pairs of 
eigenvalues of unitary random matrices.  Further evidence of a connection 
between number theory and random matrix theory was given when Keating and 
Snaith \cite{KeatingS:2000,KeatingS:2000b} used results for the 
characteristic polynomial of a random unitary matrix to formulate conjectures 
about moments of the zeta function that are supported by number-theoretic and 
numerical results (see, for instance, the review articles \cite{Keating2005, Snaith:2010, KeatingS:2011, Keating2017}).
Since then, a number of more general results on the joint moments of the 
characteristic polynomial and its derivative have been proven and used to 
formulate conjectures about the joint moments of $L$-functions and their 
derivatives 
\cite{ConreyFKRS:2003,ConreyFKRS:2005,ConreyFKRS:2008,ConreyRS:2006,Dehaye:2008,HughesKO:2000,HughesKO:2001,Winn:2012}.

Our objective is to connect the joint moments of characteristic 
polynomials with the theory of Painlev\'e equations (see also 
\cite{BornemannFM:2017} for a recent result connecting Painlev\'e functions 
with random matrix spacing distributions related to zeta zeros in a different 
direction).  Solutions to the  
Painlev\'e equations play an important role in many aspects of random matrix 
theory (see, for instance, \cite{ForresterW:2015,Its:2011}) and are  
amenable to asymptotic analysis \cite{DeiftZ:1995,FokasIKN:2006}.   

The results presented here relate $F_{N}(h,k)$ to a particular solution of the $\sigma$-Painlev\'e V equation and the  limiting function \eqref{Fdef}
to a particular solution of the $\sigma$-Painlev\'e III$^\prime$ equation.
 We see them as the first step in a longer-term project to use 
asymptotic analysis of Painlev\'e functions and related objects from the 
theory of integrable systems to obtain stronger asymptotic results on 
joint moments of characteristic polynomials.

\subsection{Results}
Let $L_n^{(\alpha)}(s)$ be the generalised Laguerre polynomial
\begin{equation}\label{laugerre1}
L^{(\alpha)}_n(s) := \frac{e^s}{s^{\alpha}n!}\frac{d^n}{ds^n}\Bigl(s^{\alpha +n}e^{-s}\Bigr) = \sum_{j=0}^{n}\frac{\Gamma(n+\alpha +1)}{\Gamma(j+\alpha +1)(n-j)!}\frac{(-s)^j}{j!}
\end{equation}
and define 
\eq\label{jc4}
K_n(\ep, y):=\frac {(-1)^n}{\pi} \frac{\partial^n}{\partial\ep ^n}\left(\frac{\ep}{\ep^2+y^2}\right).
\endeq
As we will explain in \S\ref{sec-int-rep}, $F_N(h,k)$ is related to the 
generalised Laguerre polynomials by 
\begin{equation}\label{FNint2}
F_N(h,k)=\lim_{\ep \rightarrow 0} (-1)^\frac{k(k-1)}{2}2^{-2h}\int_{-\infty}^{\infty}
K_{2h} (\ep,y)e^{{-N|y|} }\det\left[L^{(2k-1)}_{N+k-1-(i+j)}(-2|y|)\right]_{i,j=0,\dots,k-1} dy,
\end{equation}
with $N>k-1$.
Our main theorem relates the Laguerre determinant to a specific solution of 
the equation
\beq\label{painleveV}
\begin{split}
\left ( x\frac{d^2\sigma}{d x^2}\right) ^2= & \left ( \sigma -x\frac {d\sigma }{dx}+
2\left ( \frac{d \sigma}{dx}\right )^2 -2N\frac {d\sigma }{dx} \right ) ^2 \\
& -4\frac {d\sigma }{dx}\left (-N+\frac {d\sigma }{dx} \right) \left (-k-N +\frac {d\sigma }{dx}\right)\left(k+\frac {d\sigma }{dx}\right).
\end{split}
\eeq
Equation \eqref{painleveV} is a special case of the $\sigma$-Painlev\'e 
V equation with three parameters given in \eqref{sigma_jeq} (see Okamoto 
\cite{Okamoto:1980a,Okamoto:1980b,Okamoto:1987} and Jimbo and 
Miwa \cite{JM}).  For further examples of solutions of Painlev\'e equations 
expressed as Wronskian determinants of generalised Laguerre polynomials or 
confluent hypergeometric functions see 
\cite{ForresterW:2002,KM,Masuda:2004,NoumiY:1998}.

\begin{theorem}
\label{thm-PainleveV-Laguerre}
We have the representation 
\beq\label{Sum1}
\det \left[L_{N+k-1-(i+j)}^{(2k-1)} (-2|y|)\right] _{i,j=0,\cdots ,k-1}
=\frac{e^{-2k|y|}}{(2\pi i)^k} H_k[w_{0}],
\eeq
where $H_k [w_{0} ]=H_ n [w_{0} ]|_{n=k}$, and $H_n[w_0]$ is the Hankel determinant
\beq\label{Sum3}
H_n[w_0] := {\mbox{det}}\left[\int_C w_{0} (s) s^{i+j} ds \right] _{i,j=0,\cdots,n-1}
\eeq
with the weight
\beq\label{Sum2}
w_{0}(s) := \frac{e^{\frac{x}{1-s}}}{(1-s)^{2k} s^{N+k}},\;\; x=2|y|.
\eeq
Here $C$ is a small  (radius less than 1) positively oriented  circle around zero.  Furthermore, 
\beq\label{Sum4}
\frac{d}{d x} \log H_k=\frac{\sigma(x)+kx+Nk}{x},
\eeq
where $\sigma (x) $ is a  solution of the $\sigma$-Painlev\'e V equation \eqref{painleveV} with asymptotics
\beq
\label{s_asym}
\sigma(x)=-Nk + \frac{N}{2}x + \mathcal{O}(x^2), \quad x\rightarrow 0.
\eeq
\end{theorem}
Theorem \ref{thm-PainleveV-Laguerre} is proven in 
\S\ref{sec-PainleveV-Laguerre}.

The  solution of the $\sigma$-Painlev\'e V equation considered here is a rational solution. Rational solutions of Painlev\'e equations have been obtained   in \cite{Clarkson2005,Clarkson2013,Masuda} in terms of Wronksian determinants of confluent hypergeometric functions which include also the case of generalised Laguerre polynomials.
 The relation between the Wronksian determinant  in \cite{Clarkson2013} and our Hankel determinant formula \eqref{Sum3} is not immediate.  For this reason 
 we here provide an alternate proof that is quite straightforward and algorithmic to show that such determinant is a particular solution of the  Painlev\'e V equation.
 Furthermore, we show that $F_N(h,k)$ can be evaluated recursively from equation \eqref{painleveV} for integer values of $h$, giving formulae that extend to all $k$.
Using the conformal block expansion of the  $\tau$ function of the Painlev\'e V  equation introduced by Lisovyy,  Nagoya, and Roussillon \cite{Lisovyy}, we give a combinatorial expression of the coefficients $F_N(h,k)$ (see~\S\ref{conformal}).

Our second main result concerns the evaluation of the function $F(h,k)$ in 
\eqref{Fdef} for $h, k \in {\mathbb N}$, $k > h - 1/2$.  
Let $\xi(t)$ be the particular solution of the equation
\begin{equation}\label{sl40}
\begin{aligned}
\left(t\frac{d^2\xi}{dt^2}\right)^2&=- 4t\left(\frac{d\xi}{dt}\right)^3+\left(4k^2+4\xi\right)\left(\frac{d\xi}{dt}\right)^2 + t\frac{d\xi}{dt}-\xi,
\end{aligned}
\end{equation}
with initial conditions
\begin{equation}
\xi(0)=0\,,\quad \xi'(0)=0
\end{equation}
where prime denotes derivative with respect to $t$.
Equation \eqref{sl40} is a special case of the $\sg$-Painlev\'e III$^\prime$ equation 
with two parameters (cf. \eqref{sigma-Painleve-III}).  For background on 
the $\sigma$-Painlev\'e III$^\prime$ equation see Okamoto 
\cite{Okamoto:1980a,Okamoto:1980b,Okamoto:1987b} and Jimbo and Miwa \cite{JM}.
We show in Theorem~\ref{Theorem2}, which we state in \S\ref{CTM}, that 
\begin{equation}
F(h,k)
=(-1)^{h}\,\dfrac{G(k+1)^2}{G(2k+1)}\frac{d^{2h}}{dt^{2h}}
\left[\exp\int_{0}^t\left(\frac{\xi(s)}{s}ds\right)\right]\Bigg|_{t=0},
\end{equation}
where  $G$ is the Barnes function (see Appendix~\ref{Barnes}).
Furthermore,  introducing the $\tau$-function of the  Painlev\'e III    equation  defined as  
$t\dfrac{d}{dt}\log \tau_{III}(t)=\xi(t)$,  we have 
\[
F(h,k)=(-1)^{h}\frac{d^{2h}}{dt^{2h}}\tau_{III}(t)|_{t=0}.
\]
Using the conformal block expansion of $\tau_{III}(t)$   near $t=0$ we arrive at the conjectural  expression
\begin{equation}
F(h,k)=(-1)^{h}\dfrac{G(k+1)^2}{G(2k+1)}(2h)!\sum\limits_{\overset{\lambda\in\mathbb{Y}}{|\lambda|=2h,\lambda_1\leq k}}\prod_{(i,j)\in\lambda}
 \frac{ \left(2k+i-j\right)(k+i-j)}{
 h_{\lambda}^2(i,j)\left(\lambda'_j-i-j+1+2k\right)^2},\quad k>h-\frac{1}{2},
 \end{equation}
where the  sum is taken over all   Young diagrams   with  $2h$  boxes   and first row with at most $k$ boxes,  $ h_{\lambda}(i,j)$ is the hook length of the box $(i,j)$ associated to the diagram  $\lambda$, and 
 $\lambda_j'$  is the number of boxes  in the $j$  column of the   transpose diagram $\lambda'$   (see \S\ref{conformal}).
The values of the above coefficients for small values of $h$  are consistent with the results obtained  by Dehaye in \cite{Dehaye:2008}.

The connection between moments of characteristic polynomials of the unitary group and the Painlev\'e III$^\prime$ equation has already appeared in the literature.
In particular, defining  the 
characteristic polynomial of  the unitary matrix $A$ to be 
\eq
\Lambda_A(s):=\det[I-s A^*] = \prod_{n=1}^N(1-se^{-i\theta_n}),
\endeq
Conrey, Rubinstein, and Snaith proved the following.

\vspace{.1in}

\noindent
{\bf Previous result} (Conrey, Rubinstein, and Snaith \cite{ConreyRS:2006}).  
{\it For fixed $k$ and $N\rightarrow\infty$,}
\eq
\int_{U(N)}|\Lambda_A'(1)|^{2k}dA_N = b_k N^{k^2+2k} + \mathcal{O}\left(N^{k^2+2k-1}\right),
\endeq
{\it where}
\eq
\label{crs-result}
b_k:=(-1)^{k(k+1)/2}\sum_{h=0}^k\begin{pmatrix} k \\ h \end{pmatrix} \left.\left(\frac{d}{dt}\right)^{k+h}\left(e^{-t}t^{-k^2/2}\det_{k\times k}[I_{i+j-1}(2\sqrt{t})]\right)\right\vert_{t=0}.
\endeq
{\it Here $I_\nu(z)$ is the modified Bessel function of the first kind.}

\vspace{.1in}

Forrester and Witte \cite{ForresterW:2006} then showed that 

\vspace{.1in}

\eq
b_k = \frac{(-1)^k}{k!}\prod_{j=1}^k\frac{j!}{(j+k-1)!}\sum_{h=0}^k\begin{pmatrix} k \\ h \end{pmatrix} \left.\left(\frac{d}{dt}\right)^{k+h} \tau_k(t) \right\vert_{t=0},
\endeq
where $\tau_k(t)$ is a $\tau$ function of the Painlev\'e III$^\prime$ 
equation.

In work independent of ours, the method of \cite{ConreyRS:2006} has recently been extended to the joint moments of $\Lambda_A(s)$ in \cite{BBBCPRS}, yielding a result equivalent to our Theorem 2.

\subsection{Next steps}
As mentioned already, we see our result as providing a new starting point from which one can attempt to derive large-$N$ asymptotics and so seek to prove (\ref{Fdef}) and evaluate $F(h, k)$ for general $h$ and $k$.  In  \S\ref{conformal}
 we present the conformal block expansion of  the  $\tau$ function of the Painlev\'e  V   equation   \cite{Lisovyy}  and Painlev\'e III equation \cite{Lisovyy1}. Such expansions  could be potentially used to obtained formulae for $F_{N}(h,k)$ and $F(h,k)$   beyond integer values of $k$.

\

\noindent
{\bf Acknowledgements.}  The authors thank the American Institute of 
Mathematics for hospitality during the February 2017 workshop on 
\emph{Painlev\'e Equations and Their Applications}, at which this project 
began.  We thank Oleg Lisovyy for providing the Mathematica code for 
calculating the conformal blocks of the Painlev\'e equations, Peter Clarkson 
for useful discussions, and the anonymous referees for helpful 
comments.
P. Bleher was supported by NSF Grant DMS-1565602.
R. Buckingham was supported by the Charles Phelps Taft Research Center by a 
Faculty Release Fellowship and by NSF grant 
DMS-1615718.  
T. Grava was supported by H2020-MSCA-RISE-2017 PROJECT No. 778010 IPADEGAN. 
A. Its was supported by NSF Grant DMS-1700261 and Russian Science Foundation grant No.17-11-01126.
E. Its was supported by  Russian Science Foundation grant No.17-11-01126.
J.P. Keating was supported by a Royal Society Wolfson Research Merit Award, EPSRC Programme Grant EP/K034383/1
LMF: $L$-Functions and Modular Forms, and ERC Advanced Grant 740900 (LogCorRM).

\section{Integral representations for \texorpdfstring{$F_{N}(h,k)$}{TEXT}}
\label{sec-int-rep}

Our starting formulae are integral representations for $F_{N}(h, k)$ obtained in \cite{Winn:2012} for integer $k$ and integer or half-integer $h$.
The first formula involves an $(N+1)$-fold integral. 

\noindent
{\bf Proposition 1} (Proposition 1 of \cite{Winn:2012}). Let  $n \in {\mathbb N}_0 \equiv {\mathbb N} \cup \{0\}$, and define $K_n(\epsilon,y)$ by 
\eqref{jc4}.
Then, if $2h \in {\mathbb N}_0$ and $k > h -\frac{1}{2}$,
\begin{equation}\label{FNint1}
F_{N}(h,k) =\lim_{\ep \rightarrow 0} \frac{2^{N^2+2kN-2h}}{(2\pi )^N N!}
\int_{-\infty}^{\infty}\int_{-\infty}^{\infty}...\int_{-\infty}^{\infty}
K_{2h}(\ep, y) \prod_{j=1}^{N}\frac{e^{iyx_j}}{(1+x^2_j)^{N+k}}\Delta^2({\bf x})d{\bf x}dy,
\end{equation}
where $\Delta^2({\bf x}) \equiv  \Delta^2(x_1, ..., x_N)$ is the Vandermond determinant,
$$
\Delta^2({\bf x}) := \prod_{1\leq j < k \leq N}(x_k - x_j).
$$
We emphasise that formula (\ref{FNint1}) holds for any real 
$k > h -\frac{1}{2}$. As is also shown in \cite{Winn:2012}, in the case of 
integer $k$
the ${\bf x}$-integral  in the right-hand side of (\ref{FNint1}) can
be evaluated in terms of Laguerre  polynomials.  Put
$$
H(k,y) : = \int_{-\infty}^{\infty}...\int_{-\infty}^{\infty} \prod_{j=1}^{N}\frac{e^{iyx_j}}{(1+x^2_j)^{N+k}}\Delta^2({\bf x})d{\bf x}.
$$
{\bf Proposition 2} (Proposition 4 of \cite{Winn:2012}). For $k \in {\mathbb N}$ and $y \in {\mathbb R}$,
\begin{equation}\label{Hint}
H(k,y) = 
(-1)^\frac{k(k-1)}{2} \frac{(2\pi)^N N!}{2^{2kN+N^2}}e^{-N|y|} \det\left[L^{(2k-1)}_{N+k-1-(i+j)}(-2|y|)\right]_{i,j=0,\dots,k-1}.
\end{equation}
This in turn implies that for integer $k > h - \frac{1}{2}$ the $(N+1)$-fold  integral representation (\ref{FNint1}) for $F_N(h,k)$ 
can be transformed into the single integral formula \eqref{FNint2}.
This equation is our starting formula. It was also the starting point of \cite{Winn:2012}, which, for example, computed $F_{N}(h,k)$ explicitly in the case $h=\frac{1}{2} $ and $k=1$, proving the limit in (\ref{Fdef}) exists and showing that
\eq\label{jc22}
F\left(\frac{1}{2},1\right)=\frac{e^2-5}{4\pi}.
\endeq

\section {Painlev\'e V and the Laguerre determinant}
\label{sec-PainleveV-Laguerre}
In this section we  first introduce the Painlev\'e V equation and its $\tau$ function.
The goal of the section is to show that the determinant  \eqref{Sum3}  corresponds to a particular solution of the Painlev\'e V equation.

\subsection{The Painlev\'e V equation and its \texorpdfstring{$\tau$}{TEXT}-function}
In this section we summarise the main properties of the Painlev\'e V equation that can be found in \cite{JM}.
The general  Painlev\'e V  equation for a complex function $y=y(x)$ takes  the form 
\begin{equation}\label{PVy}
   \frac{d^2y}{dx^2} = \Big( \frac{1}{2y} + \frac{1}{y-1} \Big) \left(\frac{dy}{dx}\right)^2
         - \frac{1}{x} \frac{dy}{dx}
         + \frac{(y-1)^2}{x^2} \Big( \alpha y + \frac{\beta}{y} \Big) 
         + \gamma\frac{y}{x} + \delta\frac{y(y+1)}{y-1}.
\end{equation}
The coefficients $\alpha$,  $\beta$, $\gamma$, and $\delta$ are complex constants.
One can fix   $\delta = - \frac{1}{2}$ because 
the general Painlev\'e V equation with $\delta \ne 0$ can be reduced to the case
with $\delta = - \frac{1}{2}$ by the mapping $x \mapsto \sqrt{-2 \delta} x$.

The equation \eqref{PVy}  has a Lax pair, namely  it  can be written as the compatibility condition
of two linear systems  of ODEs for the $2\times 2 $ matrix function $\Phi(z,x)$, $z,x,\in\mathbb{C}$, that satisfies the equations
\begin{align}
\label{Lax1}
\dfrac{\partial\Phi}{\partial z}&=\left(\dfrac{x}{2}\sigma_3+\dfrac{A_0}{z}+\dfrac{A_1}{z-1}\right)\Phi(z,x), \quad \sigma_3:=\begin{pmatrix} 1 & 0 \\ 0 & -1 \end{pmatrix},\\
\label{Lax2}
\dfrac{\partial\Phi}{\partial x}&=\left(\dfrac{z}{2}\sigma_3+\dfrac{B_0}{x}\right)\Phi(z,x),
\end{align}
with
\beq\label{TP21}
A_0 := \begin{pmatrix}
w+\frac{\theta_0}{2}& -u(w+\theta_0) \\
u^{-1}w & -w-\frac{\theta_0}{2}
\end{pmatrix},
\eeq
\beq\label{TP23}
A_1:=\begin{pmatrix}
-w-\frac{\theta_0+\theta_\infty}{2} & uy \left( w+\frac{\theta_0-\theta_1+\theta_\infty}{2}\right) \\
-(uy)^{-1}\left(w+\frac{\theta_0+\theta_1+\theta_\infty}{2}\right) & w+\frac{\theta_0+\theta_\infty}{2}
\end{pmatrix},
\eeq
and 
\begin{equation}\label{B0}
B_0 := 
\begin{pmatrix}
0& -u(w+\theta_0)+uy \left( w+\frac{\theta_0-\theta_1+\theta_\infty}{2}\right) \\
u^{-1}w-(uy)^{-1}\left(w+\frac{\theta_0+\theta_1+\theta_\infty}{2}\right) & 0
\end{pmatrix}\frac{1}{x},
\end{equation}
where $\theta_j$, $j=0,1,\infty$ are  constant  parameters and $u\equiv u(x)$, $w\equiv w(x)$, and $y\equiv y(x)$.
As shown in \cite{JM}, the compatibility condition     of equations \eqref{Lax1} and \eqref{Lax2}, namely
\begin{equation}\label{Compatibility}
\dfrac{\partial}{\partial x}\dfrac{\partial\Phi}{\partial z}=\dfrac{\partial}{\partial z}\dfrac{\partial\Phi}{\partial x},
\end{equation}
 implies that the functions  $w, y, u$ 
 satisfy the following  $3\times3$ system of  first-order ordinary differential equations (cf. \cite[(C.40)]{JM}):
\beq\label{dydt}
x\frac{dy}{dx} = xy -2w(y-1)^2 -(y-1)\left(\frac{\theta_0-\theta_1+\theta_{\infty}}{2}y - \frac{3\theta_0+\theta_1+\theta_{\infty}}{2}\right),
\eeq
\beq\label{dxdt}
x\frac{dw}{dx} = yw\left(w+ \frac{\theta_0-\theta_1+\theta_{\infty}}{2}\right) -\frac{1}{y}(w+\theta_0)
\left(w+ \frac{\theta_0+\theta_1+\theta_{\infty}}{2}\right),
\eeq
\beq\label{dudt}
x\frac{d}{dx}\log u = -2w -\theta_0 +y\left(w+ \frac{\theta_0-\theta_1+\theta_{\infty}}{2}\right) +   
\frac{1}{y}
\left(w+ \frac{\theta_0+\theta_1+\theta_{\infty}}{2}\right).
\eeq
Equation (\ref{dudt}) just gives $u$ in terms of $y$ and $w$, while the system of two equations (\ref{dydt})--(\ref{dxdt}) is, in fact, 
equivalent to the fifth Painlev\'e equation \eqref{PVy}  for the function $y(x)$,
where
\beq\label{alpha}
\alpha= \frac{1}{2}\left(\frac{\theta_0-\theta_1 + \theta_{\infty}}{2}\right)^2,\quad \beta= -\frac{1}{2}\left(\frac{\theta_0-\theta_1 - \theta_{\infty}}{2}\right)^2, \quad \gamma= 1-\theta_0 -\theta_1, \quad \delta =- \frac{1}{2}.
\eeq

We observe that the matrices $A_0$, $A_1$, and $B_0$ in \eqref{TP21}, \eqref{TP23}, and \eqref{B0} are traceless and the eigenvalues of $A_0$ and $A_1$ are constants:
\beq\label{TP19}
{\mbox{Spect}}\, A_0=\left\{ \pm \frac{\theta_0}{2}\right\}
\eeq
and
\beq\label{TP20}
{\mbox{Spect}} \,A_1=\left\{ \pm \frac{\theta_1}{2}\right\}.
\eeq
It follows that 
  the solution  of  equation \eqref{Lax1}    in the neighbourhood of the   regular singular  points $z=0$ and $z=1$ takes the form
\beq\label{TI0}
\Phi(z)=  \widehat{\Phi}_0(z) z^{\frac{\theta_0}{2}\sigma_3},
\quad z\sim 0,
\eeq

\beq\label{TI1}
\Phi(z)=\widehat{\Phi}_{1}(z)( z-1)^{\frac{\theta_1}{2}\sigma_3},\quad z \sim 1,
\eeq
where $\widehat{\Phi}_0 (z)$ and $\widehat{\Phi}_1 (z)$ are 
holomorphic and invertible in neighbourhoods of the respective points. Regarding the behaviour of the solution of \eqref{Lax1} near the irregular singular point $z=\infty$ of Poincar\'e rank $1$, one has to consider the diagonal part, $\mbox{diag}(A_0+A_1)=-\frac{\theta_{\infty}}{2}\sigma_3$.  It follows that the formal solution  of $\Phi(z)$  in the neighbourhood of $z=\infty$  takes the form  \cite{JM}
\beq\label{TIin}
\Phi(z)=\widehat{\Phi}_{\infty}(z) e^{\frac{xz}{2}\sigma _3} z^{\frac{\theta_{\infty}}{2}\sigma_3},\quad z\sim \infty,
\eeq
where $\widehat{\Phi}_\infty(z)$ 
has the formal asymptotic expansion  at $z=\infty$
\beq\label{Phihatinfty0}
\widehat{\Phi}_{\infty}(z) = I + \frac{\phi_1}{z} +\dots,\quad \quad z\sim \infty,
\eeq
where $\phi_1$ is a matrix independent of $z$. We define the Hamiltonian $ \mathscr{H}$ as
\begin{equation}\label{HamP5}
\mathscr{H}:=-(\phi_1)_{11},
\end{equation}
where $(\phi_1)_{11}$ denotes the $11$ entry of the matrix $\phi_1$. In fact, one has
\begin{equation}
\label{phi_H}
\mathscr{H}:=-\dfrac{1}{x}\left(w-\dfrac{1}{y}\left(w+\frac{\theta_0+\theta_1+\theta_{\infty}}{2}\right)\right)
\left(w+\theta_0-y\left(w+\frac{\theta_0-\theta_1+\theta_{\infty}}{2}\right)\right)-w-\dfrac{\theta_0+\theta_\infty}{2}.
\end{equation}
Up to linear terms in $x$, the function $x \mathscr{H}$ satisfies a second-order ODE which is called the 
$\sigma$-form of the Painlev\'e V equation.
More precisely, defining the function 
\begin{equation}
\label{sigma}
\sigma:=x \mathscr{H}+\dfrac{1}{2}(\theta_0+\theta_{\infty})x+\dfrac{1}{4}(\theta_0+\theta_{\infty})^2-\dfrac{\theta_1^2}{4},
\end{equation}
then 

%
%
%
\begin{multline}\label{sigma_jeq}
  \left(x\frac{d^2\sigma}{dx^2}\right)^2 = \left[ \sigma - x \frac{d\sigma}{dx} + 2\left(\frac{d\sigma}{dx}\right)^2 
  -(2\theta_0+\theta_{\infty}) \frac{d\sigma}{dx} \right]^2 \\
  - 4 \frac{d\sigma}{dx}\left( \frac{d\sigma}{dx}-\theta_0\right)\left(\frac{d\sigma}{dx}-\frac{\theta_0-\theta_1+\theta_{\infty}}{2}\right)
  \left( \frac{d\sigma}{dx}-\theta_0-\frac{\theta_0+\theta_1+\theta_{\infty}}{2}\right) .
\end{multline}


%
%
Finally, the $\tau$-function is defined (cf. \cite{JM, JMU}) in terms of the 
Hamiltonian  $\mathscr{H}$  by
\begin{equation}\label{tau}
  \mathscr{H} =: \frac{d}{dx} \log \tau,
\end{equation}
so that, by \eqref{sigma},
\[
\sigma=x \dfrac{d}{dx}\log\left (x^{(\frac{1}{4}(\theta_0+\theta_{\infty})^2-\frac{1}{4}\theta_1^2)  }e^{\frac{1}{2}(\theta_0+\theta_{\infty})x}\tau(x)\right).
\]

In the next subsection we show that the Hankel determinant \eqref{Sum3} is  a  tau-function of the 
Painlev\'e V equation. We will proceed as follows:
\begin{itemize}
\item[1.] We  formulate a Riemann-Hilbert problem for the generalised Laguerre polynomials \eqref{laugerre1} and we derive a system of ODEs related to this Riemann-Hilbert problem;
\item[2.] We  introduce a series of rational and gauge transformations to reduce the   system of ODEs  obtained   in 1. to the Lax pair \eqref{Lax1}--\eqref{Lax2};
\item[3.] Finally, we identify the   Hankel determinant \eqref{Sum3} with  the $\tau$-function of the Painlev\'e V equation via the relation \eqref{sigma}.
\end{itemize}
\subsection{Laguerre determinant}
We will analyse  the principal ingredient of the starting formula (\ref{FNint2}), i.e. the 
{\it Laguerre determinant}
\beq\label{PV2}
\det\left[L_{N+k-1-(i+j)}^{(2k-1)} (-2|y|)\right]_{i,j=0,\cdots , k-1}.
\eeq

Recall the classical formula 
\beq\label{PV1}
L_n^{(\alpha)}(x)=\int_C\frac{e^{-xt/(1-t)}}{(1-t)^{\alpha+1}t^{n+1}}dt
\eeq
for the generalised Laguerre polynomials, where C is a closed contour around 
$0$ (for instance, the positively oriented circle around $0$
with radius $\frac{1}{2}$).  In our case,  we have 
\beq\label{PV3}
\begin{split}
L_{N+k-1-(i+j)}^{(2k-1)} (x) & = \frac{1}{2\pi i}\int_C\frac{e^{-\frac{xt}{1-t}}}{(1-t)^{2k}
t^{N+k-(i+j)}}dt \\
  & =\frac{1}{2\pi i}\int_C\frac{e^{-\frac{xt}{1-t}}}{(1-t)^{2k}
t^{N+k}}t^{i+j}dt =\frac{1}{2\pi i}e^x\int_C\frac{e^{-\frac{x}{1-t}}}{(1-t)^{2k}
t^{N+k}}t^{i+j}dt.
\end{split}
\eeq
Hence
\beq\label{PV5}
\det\left[L_{N+k-1-(i+j)}^{(2k-1)} (-2|y|)\right]_{i,j=0,\cdots , k-1}=\frac{e^{-2k|y|}}
{(2\pi i)^k} H_k[w_0],
\eeq
where
\beq\label{hankeldef}
H_{n}[w_0]:= \det\left[\int_Ct^{i+j}w_0(t)dt\right]_{i,j = 0,\dots,n-1}
\eeq
is the Hankel determinant with the weight

\beq\label{PV6}
w_0(t)=\frac{e^{2|y|/(1-t)}}{(1-t)^{2k}t^{N+k}}.  
\eeq
Following the general Riemann-Hilbert  scheme in the theory of Hankel  
determinants (see e.g. \cite{BleherIts:1999,its}), we consider the system of monic orthogonal polynomials
with weight $w_0(t)$ on $C$,
$$
P_n(t) = t^n + \dots,  \quad \quad\int_{C}P_n(t)t^{m}w_0(t)dt = h_n\delta_{nm}, \quad m = 0,\dots,  n, 
$$
and define the $2\times2$ matrix valued function
\beq\label{Ydef0}
Y(t) := \begin{pmatrix} P_n(t) & \frac{1}{2\pi i} \int_{C}\frac{P_n(t')w_0(t')dt'}{t'-t}\\\\
-\frac{2\pi i}{h_{n-1}}P_{n-1}(t)& -\frac{1}{h_{n-1}}\int_{C}\frac{P_{n-1}(t')w_0(t')dt'}{t'-t}
\end{pmatrix}.
\eeq
The defining property of the function $Y(t)$ is that it is the unique solution of the 
following matrix Riemann-Hilbert problem:

\beq\label{PV9}
Y(t) \in \mathcal{H} ({\mathbb C} \setminus C ),
\eeq

\beq\label{PV10}
Y_+(t)=Y_-(t) \begin{pmatrix}
1 & \frac{e^{\frac{x}{1-t}}}{(1-t)^{2k}t^{N+k}} \\
0 & 1
\end{pmatrix},
\quad t \in C,\quad x = 2|y|,
\eeq
\beq\label{PV11}
Y(t) =  \left(I + \mathcal{O}\left(\frac{1}{t}\right)  \right )t^{n\sigma_3},  \quad t \rightarrow \infty, 
\eeq
where $ { \mathcal H} ({\mathbb C} \setminus C )$ stands for the holomorphic 
$2\times 2$ matrix-valued functions in ${\mathbb C} \setminus C$.
The Hankel determinant $H_n[w_0]$ is related to the function $Y(t)$ via the equations
\beq\label{PV28}
\frac{H_{n+1}}{H_n}=h_n,\quad h_n=-2\pi i (m_1)_{12},\quad \frac{1}{h_{n-1}}=-\frac{1}{2\pi i}(m_1)_{21},
\eeq
where the matrix $m_1$ is  the first coefficient in the expansion (\ref{PV11}), i.e.,
\beq\label{PV110}
Y(t) \sim  \left(I + \frac{m_1}{t} + \cdots \right )t^{n\sigma_3},  \quad t \rightarrow \infty.
\eeq
It also should be noticed  that 
\beq
\det Y(t) \equiv 1.
\eeq

Let us change the variable $t$ to $z$ : $z=\frac{1}{1-t}$, $t=\frac{z-1}{z}$ so that the circle  C  maps to the 
new circle
\beq
\Gamma := \left\{ z: \left|z-\frac{4}{3}\right| = \frac{2}{3}\right\},
\eeq
oriented counterclockwise. Put
\beq\label{PV12}
X(z):=Y(t(z))\equiv Y\left(\frac{z-1}{z}\right).
\eeq
Then, in terms of the function $X(z)$, the Riemann-Hilbert problem  (\ref{PV9})--(\ref{PV11}) reads as follows:
\beq\label{PV13}
X(z) \in \mathcal{H} ({\mathbb C}P^1 \setminus (\Gamma\cup\{0\})),
\eeq
\beq\label{PV14}
X_+(z)=X_-(z) \begin{pmatrix}
1 & e^{xz}\frac{z^{3k+N}}{(z-1)^{N+k}} \\
0 & 1
\end{pmatrix},
\quad z \in \Gamma,
\eeq
\beq\label{PV15}
X(z) \sim \left(I + m^{X}_1z + \cdots \right )(-z)^{-n\sigma_3},  \quad z \rightarrow 0,
\eeq
where 
\beq\label{m1X}
m^{X}_1 := -m_1 -n\sigma_3,
\eeq
and also
\begin{equation}\label{detX}
\det X(t) \equiv 1.
\end{equation}
From the point of view of the modern theory of isomonodromic deformations, the Riemann-Hilbert 
problem (\ref{PV12})--(\ref{PV15}) indicates that we are dealing with the fifth Painlev\'e equation. In what follows we present the detailed derivation of this fact.

Define the function
\beq\label{PV16}
\Psi (z) :=X(z) \begin{pmatrix}
e^{\frac{xz}{2}} & 0 \\
0 & e^{-\frac{xz}{2}}\frac{(z-1)^{k+N}}{z^{N+3k}}
\end{pmatrix}.
\eeq
Then,

\beq\label{PV17}
\Psi(z) , \Psi^{-1}(z) \in {\mathcal H}  (\mathbb{C} \setminus ( \Gamma \cup \{0\} \cup\{1\} ))
\eeq
and the function $\Psi(z)$  has a constant jump across the circle $\Gamma$,
\beq\label{PV18}
\Psi_+(z)=\Psi_-(z)\begin{pmatrix}
1 & 1 \\
0 & 1
\end{pmatrix},
\quad z \in \Gamma.
\eeq
Moreover, in neighbourhoods of the points $\infty$, $0$, and $1$, the function $\Psi(z)$ exhibits the following behaviour:

\beq\label{PV22}
\Psi(z)=\widehat{\Psi}_\infty (z) \begin{pmatrix}
e^{\frac{xz}{2}} & 0 \\
0 & e^{-\frac{xz}{2} }
\end{pmatrix}
\begin{pmatrix}
1 & 0 \\
0 & z^{-2k}
\end{pmatrix},
\quad z \sim \infty,
\eeq

\beq\label{PV20}
\Psi(z)=\widehat{\Psi}_0(z)\begin{pmatrix}
z^{-n} & 0 \\
0 & z^{n-3k-N}
\end{pmatrix},
\quad z \sim 0,
\eeq

\beq\label{PV19}
\Psi(z)=\widehat{\Psi}_1(z)\begin{pmatrix}
1 & 0 \\
0 & (z-1)^{N+k}
\end{pmatrix},
\quad z \sim 1,
\eeq
where $\widehat{\Psi}_\infty (z)$, $\widehat{\Psi}_0 (z)$, and $\widehat{\Psi}_1 (z)$ are 
holomorphic and invertible in the neighbourhoods of the respective points. Also,
\beq\label{PV210}
\widehat{\Psi}_{\infty}(\infty) = X(\infty) = Y(1),
\eeq
\beq\label{PV21}
\widehat{\Psi}_0(0)= (-1)^{-n\sigma _3} \begin{pmatrix}
1 & 0 \\
0 & (-1)^{k+N}
\end{pmatrix} = \begin{pmatrix}
(-1)^{-n} & 0 \\
0 & (-1)^{k+N+n}
\end{pmatrix},
\eeq
\beq\label{PV211}
\widehat{\Psi}_{1}(1) = X(1)e^{\frac{x}{2}\sigma_3} = Y(0)e^{\frac{x}{2}\sigma_3}.
\eeq
We should mention that the invertibility of the functional factors $\widehat{\Psi}_{j}(z)$, $j = 0, 1, \infty$ follows 
from the equation
\beq
\det \Psi(z) = \frac{(z-1)^{k+N}}{z^{N+3k}},
\eeq
which in turn is a consequence of (\ref{detX}).

By standard arguments, the properties (\ref{PV17})--(\ref{PV19}) imply that the function $\Psi(z) \equiv \Psi(z,x)$ 
satisfies linear differential equations with respect to $z$ and $x$ of the form
\beq\label{PV23}
\frac{\partial\Psi}{\partial\,z}=\left( x\widehat{A}_\infty +\frac{\widehat{A}_0}{z}+\frac{\widehat{A}_1}{z-1} \right)\Psi,
\eeq
\beq\label{PV24}
\frac{\partial\Psi}{\partial\,x}=\left( z\widehat{B}_\infty +\widehat{B}_0\right)\Psi,
\eeq
where
\beq\label{PV25}
\widehat{B}_0=0,
\eeq
\beq\label{PV26}
\widehat{A}_\infty=\widehat{B}_\infty =\frac{1}{2}X(\infty)\sigma_3 X^{-1}(\infty )=\frac{1}{2}Y(1)\sigma_3 Y^{-1}(1),
\eeq
\beq\label{PV27}
\widehat{A}_0=\begin{pmatrix}
-n & 0 \\
0 &  n-3k-N
\end{pmatrix},
\eeq
\beq\label{PV36}
\widehat{A}_1=X(1)\begin{pmatrix}
0& 0 \\
0 & N+k
\end{pmatrix}
X^{-1}(1)
=Y(0)\begin{pmatrix}
0& 0 \\
0 & N+k
\end{pmatrix}
Y^{-1}(0).
\eeq
Indeed, since the jump matrix in (\ref{PV18}) is constant, the logarithmic derivatives 
$\frac{\partial\Psi}{\partial z}\Psi^{-1}(z)$ and $\frac{\partial\Psi}{\partial x}\Psi^{-1}(z)$
do not have jumps across $\Gamma$ and  hence are analytic in 
$\mathbb{C} \setminus ( \{0\} \cup\{1\} )$.  In addition, formulae (\ref{PV22})--(\ref{PV19}) tell us 
that $\frac{\partial\Psi}{\partial z}\Psi^{-1}(z)$ has simple poles at $z=0$,  $z=1$ and  is holomorphic at $z = \infty$ while
$\frac{\partial\Psi}{\partial x}\Psi^{-1}(z)$ is holomorphic at $z=0$,  $z=1$  and has a simple pole at $z=\infty$. These arguments
yield equations (\ref{PV23}) and (\ref{PV24}). Moreover, as $z \rightarrow \infty$,
\beq
\frac{\partial\Psi}{\partial z}\Psi^{-1}(z)= \frac{x}{2}\widehat{\Psi}_{\infty}(\infty)\sigma_3\widehat{\Psi}^{-1}_{\infty}(\infty) + \mathcal{O}\left(\frac{1}{z}\right)
\eeq
and 
\beq
\frac{\partial\Psi}{\partial x}\Psi^{-1}(z)= \frac{z}{2}\widehat{\Psi}_{\infty}(\infty)\sigma_3\widehat{\Psi}^{-1}_{\infty}(\infty) + \mathcal{O}\left(1\right)
\eeq
which, together with (\ref{PV210}), imply (\ref{PV26}). Similarly, as $z\rightarrow 0$,
$$
\frac{\partial\Psi}{\partial z}\Psi^{-1}(z)= \frac{1}{z}\widehat{\Psi}_{0}(0)\begin{pmatrix}-n&0\\
 0&n-3k-N\end{pmatrix} \widehat{\Psi}^{-1}_{0}(0) + \mathcal{O}\left(1\right)
$$ 
and, as $z \rightarrow 1$,
$$
\frac{\partial\Psi}{\partial z}\Psi^{-1}(z)= \frac{1}{z-1}\widehat{\Psi}_{1}(1)\begin{pmatrix}0&0\\
 0&N+k\end{pmatrix} \widehat{\Psi}^{-1}_{1}(1) + \mathcal{O}\left(1\right)
$$ 
which, together with (\ref{PV21}) and (\ref{PV211}), imply (\ref{PV27}) and (\ref{PV36}), respectively.
Finally, as $z \rightarrow 0$, we have that
$$
\frac{\partial \Psi}{\partial x}\Psi^{-1}(z)= \frac{\partial\widehat{\Psi}_0(z)}{\partial x}\widehat{\Psi}^{-1} _0(z) = 
\frac{\partial\widehat{\Psi}_0(0)}{\partial x}\widehat{\Psi}^{-1} _0(0) + \mathcal{O}(z) = \mathcal{O}(z),
$$ 
which implies equation  (\ref{PV25}).

The matrix equation  (\ref{PV23}) is a $2\times 2$ system with rational coefficients having three singular points: two Fuchsian points at $z =0$ and $z = 1$, and 
one irregular singular point with Poincar\'e index 1 at $z = \infty$. The presence of the second matrix equation (\ref{PV24}) shows that the $x$-dependence of the coefficients of system (\ref{PV23}) is monodromy-preserving. In other words, we are dealing with the 
Painlev\'e-type isomonodromy deformation of  (\ref{PV23}). In fact, the pair  of matrix equations (\ref{PV23})--(\ref{PV24}) 
 is  almost the Lax pair \eqref{Lax1}--\eqref{Lax2} for the fifth Painlev\'e  equation \eqref{PVy}
given by Jimbo-Miwa \cite{JM}. To make  it exactly the Jimbo-Miwa Painlev\'e V Lax pair a little extra work is needed. 
We observe that the matrices $\widehat{A}_1$ and $\widehat{A}_0$  are not traceless  since $$\mbox{Trace} \,\widehat{A}_0=-3k-N,\quad \mbox{Trace} \,\widehat{A}_1=N+k$$ and $\Psi(z,x)$ is not normalised to the identity at   $z=\infty$.
We make the transformation
\beq\label{TP1}
 \Phi(z,x):=z^{\frac{3k+N}{2} } (z-1)^{-\frac{N+k}{2}}\widehat{\Psi}_\infty^{-1}(\infty) \Psi(z,x)=z^{\frac{3k+N}{2} } (z-1)^{-\frac{N+k}{2}}Y^{-1}(1)\Psi(z,x)
\eeq
 that brings the original system (\ref{PV23})--(\ref{PV24})  to the normalised-at-infinity and traceless form  \eqref{Lax1}--\eqref{Lax2},
 where 
%
%

\begin{equation}\label{A0tilde}
A_0 = Y^{-1}(1)\widehat{A}_0Y(1) + \frac{3k+N}{2} I = 
Y^{-1}(1) \begin{pmatrix}
-n+\frac{3k+N}{2} & 0 \\
0 & n-\frac{3k+N}{2}
\end{pmatrix} Y(1),
\end{equation}
\beq\label{A1tilde}
A_1= Y^{-1}(1)\widehat{A}_1Y(1) - \frac{k+N}{2} I
= Y^{-1}(1)Y(0) \begin{pmatrix}
-\frac{k+N}{2} & 0 \\
0 & \frac{k+N}{2}
\end{pmatrix}  Y^{-1} (0) Y(1),
\eeq
and
\beq\label{B0tilde}
B_0 = -Y^{-1}(1)Y_x(1).
\eeq
We also notice that the tracelessness of the coefficients  of the matrix 
$B_0$ follows from the identity $\det Y(t) \equiv 1$.  To identify the matrices $A_0$ and $A_1$ in \eqref{A0tilde} and \eqref{A1tilde} with those defined in \eqref{TP21} and \eqref{TP23}
we need some extra work.
To this end we notice  that  the local equations (\ref{PV22})--(\ref{PV19}) in terms of the new function  $\Phi(z)$ read

\beq\label{TP14}
\Phi(z)=\widehat{\Phi}_{\infty}(z) e^{\frac{xz}{2}\sigma _3} z^{k\sigma_3},\quad z\sim \infty,
\eeq

\beq\label{TP15}
\Phi(z)=  \widehat{\Phi}_0(z) z^{\left(-n + \frac{3k+N}{2}\right)\sigma_3},
\quad z\sim 0,
\eeq

\beq\label{TP16}
\Phi(z)=\widehat{\Phi}_{1}(z)( z-1)^{-\frac{N+k}{2}\sigma_3},\quad z \sim 1,
\eeq
where $\widehat{\Phi}_\infty (z)$, $\widehat{\Phi}_0 (z)$, and $\widehat{\Phi}_1 (z)$ are 
holomorphic and invertible in neighbourhoods of the respective points. Also,
\beq\label{PV2100}
\widehat{\Phi}_{\infty}(\infty) = I,
\eeq
\beq
\widehat{\Phi}_0(0)= Y^{-1}(1) \begin{pmatrix}
(-1)^{-n} & 0 \\
0 & (-1)^{k+N+n}
\end{pmatrix},
\eeq
\beq\label{PV2110}
\widehat{\Phi}_{1}(1)=Y^{-1}(1) Y(0)e^{\frac{x}{2}\sigma_3}.
\eeq
Comparing (\ref{TP14})--(\ref{TP16}) with  \eqref{TI0}--\eqref{TIin}    we see that, in our case,  the formal monodromy
exponents $\theta_{\infty}$, $\theta_{0}$, and $\theta_{1}$  are 
\beq\label{TP17}
\theta_\infty=-2k ,\quad \theta_0=-2n+3k+N,\quad \theta_1=-N-k.
\eeq
Note that, simultaneously, these equations determine the diagonal part of the sum 
of the matrices $A_0$ and $A_1$ and also their  spectrums. Indeed, from (\ref{TP14})
we have that 
\beq\label{TP18}
\mbox{diag}\left(A_0+A_1\right)=k\sigma_3,
\eeq
and from \eqref{A0tilde}--\eqref{A1tilde}  or  (\ref{TP15})--(\ref{TP16})  we have that
\beq\label{TP190}
{\mbox{Spect}}\, A_0=\left\{ \pm \frac{\theta_0}{2}\right\}\equiv\left\{\pm \left(-n+\frac{3k+N}{2}\right)\right\},
\eeq
and
\beq\label{TP200}
{\mbox{Spect}} \,A_1=\left\{ \pm \frac{\theta_1}{2}\right\}\equiv\left\{\pm \left(\frac{N+k}{2}\right)\right\}.
\eeq
The last three relations mean that, with $k,n,N$ fixed, we can parameterise $A_0$ and $A_1$ by just {\it three}  parameters. We denote them $w,y,u$ and, following \cite{JM}, 
the matrices $A_0$ and $A_1$ can be parameterised in  the  form \eqref{TP21} and \eqref{TP23}.
%
Using the general identity{\footnote{ To derive this identity one substitutes 
the expansion
$$
\Phi(z)\equiv \widehat{\Phi}_{\infty}(z) e^{\frac{xz}{2}\sigma _3} z^{k\sigma_3}
=  \left( I + \frac{\phi_1}{z} + \cdots\right)e^{\frac{xz}{2}\sigma _3} z^{k\sigma_3}
$$
into the Lax pair system (\ref{Lax1})--(\ref{Lax2}).  This leads to the following  formulae for the
matrix coefficients $A_{0,1}$ and $B_{0}$ in terms of the same matrix coefficient 
$\phi_1$:
$$
A_0 +A_1 = k\sigma_3 + \frac{x}{2}[\phi_1, \sigma_3]
\quad\mbox{and}\quad B_0  = \frac{1}{2}[\phi_1, \sigma_3]. 
$$
Identity (\ref{A01B}) follows.}}
\begin{equation}\label {A01B}
A_0 + A_1 = k\sigma_3 + x B_0,
\end{equation}
we obtain the expression for $B_0$ in \eqref{B0}.

As shown in \cite{JM}, the compatibility condition of  \eqref{Lax1}--\eqref{Lax2}  
 implies that the parameters $w, y, u$  become functions of $x$ and
they satisfy the  $3\times3$ system of  first-order ordinary differential equations   \eqref{dydt}--\eqref{dudt} with parameters
 \beq\label{alpha0}
\alpha= \frac{1}{2}\left(\frac{\theta_0-\theta_1 + \theta_{\infty}}{2}\right)^2 = \frac{1}{2}(N-n +k)^2,
\eeq
\beq\label{beta0}
\beta= -\frac{1}{2}\left(\frac{\theta_0-\theta_1 - \theta_{\infty}}{2}\right)^2 = \frac{1}{2}(N-n +3k)^2,
\eeq
\beq\label{gamma0}
\gamma= 1-\theta_0 -\theta_1 = 1 +2n -2k.
\eeq

Let us now obtain the formula for the Hankel determinant $H_n[w_0]$ in terms of the functions $y(x)$ and $w(x)$.  We   use the relation \eqref{phi_H}  that in our case takes the form
\beq
\label{tildemywu}
\begin{split}
(\phi_1)_{11} & \equiv  -\mathscr{H}\\
& = \frac{1}{x}\left(w-\frac{1}{y}(w-n)\right)\Bigl(w-2n+3k +N-y(w+N +k-n)\Bigr) +w+\frac{N+k-2n}{2}.
\end{split}
\eeq
We have the following lemma.
\begin{lemma}
\label{Hn-m-lemma}
The following relation between $H_n[w_0]$ and $\phi_1$ 
holds:
\beq\label{Hm1tilde}
\frac{d}{dx}\log H_n[w_0] = -\phi_{1,11} + \frac{N+k}{2}.
\eeq
\end{lemma}
Lemma \ref{Hn-m-lemma} is proven 
in Appendix \ref{lemma-appendix}.

Combining \eqref{tildemywu} and  Lemma~\ref{Hn-m-lemma}, we arrive at the following expression for $\frac{d}{dx}\log H_n[w_0] $ in terms
of the Painlev\'e V function $y(x)$:
\beq\label{Hyw}
\frac{d}{dx}\log H_n[w_0]=\mathscr{H} + \frac{N+k}{2},
\eeq
observing that $w$   contained in $\mathscr{H}$ can be expressed as a function of $y$ and its first derivative using \eqref{dxdt}.
In what follows we will give some details of the asymptotic  expansion of $\frac{d}{dx}\log H_n[w_0]$ by introducing the  sigma function \eqref{sigma}, 
  \beq\label{sigmadef}
 \sigma_n(x):=x \frac{d}{dx}\log H_n[w_0]-nx  -n(N+k-n).
 \eeq
 The $\sigma_n$ function satisfies the equation \eqref{sigma_jeq}  with parameters \eqref{TP17},
namely\beq\label{sigmaeq1}
\begin{split}
\left ( x\frac{d^2\sigma_n}{d x^2}\right) ^2= & \left ( \sigma_n -x\frac {d\sigma_n }{dx}+
2\left ( \frac{d \sigma_n}{dx}\right )^2 +(4n-4k-2N)\frac {d\sigma_n }{dx} \right ) ^2 \\
& -4\frac {d\sigma_n }{dx}\left (n-k-N+\frac {d\sigma_n }{dx} \right) \left (2n-3k-N +\frac {d\sigma_n }{dx}\right)\left(n+\frac {d\sigma_n }{dx}\right),
\end{split}
\eeq
and for the particular case $n=k$ one obtains the equation in \eqref{painleveV}.

Recalling  formula (\ref{PV5}) for the Laguerre determinant we are studying, 
we have now arrived at the representation of the determinant in terms of a 
special solution of the fifth Painlev\'e equation (\ref{sigmaeq1}).  It 
remains to determine the $\mathcal{O}(x)$ term in the asymptotic expansion of 
$\sigma_k(x)$ as $x\to 0$. 
\begin{lemma}
\label{lemma_exp}
The solution of the $\sigma$-Painlev\'e V equation \eqref{painleveV} 
appearing in \eqref{Sum4} has the asymptotic expansion 
\eq
\label{ex_sigma}
\sigma_k(x) = -Nk + \frac{N}{2}x + \sum_{j=1}^k \alpha_{2j} x^{2j}+{\mathcal O}(x^{2k+1})\quad \mbox{as $x\to 0$},
\endeq
where the  coefficients $\alpha_{2j}$, $j=1,\dots,k$, are uniquely determined  recursively from the equation \eqref{painleveV}.
\end{lemma}
Note that the odd-power coefficients $\alpha_{2j+1}$ will generically 
be non-zero for $j\geq k$.
\begin{proof}
The function $\sigma_k$ satisfies the equation \eqref{painleveV}, namely
\begin{equation}\label{zz4}
\begin{split}
\left ( x\frac{d^2\sigma_k}{d x^2}\right) ^2 = & \left ( \sigma_k -x\frac {d\sigma_k}{dx}+
2\left ( \frac{d \sigma_k}{dx}\right )^2 -2N\frac {d\sigma_k }{dx} \right ) ^2\\ 
 &  -4\frac {d\sigma_k }{dx}\left (-N+\frac {d\sigma_k }{dx} \right) \left (-k-N +\frac {d\sigma_k }{dx}\right)\left(k+\frac {d\sigma_k }{dx}\right),
\end{split}
\end{equation}
with the initial data
\begin{equation}\label{zz5}
\sg_k(0)=-Nk,\quad \sg_k'(0)=\frac{N}{2}
\end{equation}
(we recall prime denotes the derivative with respect to $x$).
It is important to notice that equation \eqref{zz4} is degenerate at $x=0$ 
and the Cauchy-Kovalevskaya theorem is not applicable to the initial value 
problem \eqref{zz4}--\eqref{zz5}. Moreover, at $x=0$ the term with 
$\sg_k''(0)$ vanishes and
we get a relation between $\sg_k'(0)$ 
and $\sg_k(0)$. Substituting $\sg_k(0)=-Nk$ into equation \eqref{zz4} and 
assuming $k\not=0$, we obtain that $\sg_k'(0)$,
hence the second initial condition in \eqref{zz5} is automatically satisfied for any solution $\sg_k(x)$ of equation \eqref{zz4} with $\sg_k(0)=-Nk$,
$k\not =0$. 
To determine the higher-order terms in the asymptotic expansion of the function $\sigma_k(x)$ as $x\to0$,  we introduce the function
$\ts(x)$ such that 
\begin{equation}\label{zz7}
\sg_k(x)=-Nk+\frac{Nx}{2}+\widetilde\sg(x).
\end{equation}
If we substitute the above expression in \eqref{painleveV},
then we obtain the equation
\begin{equation}\label{zz8}
\left(x\frac{d^2\widetilde\sg}{dx^2}\right)^2=-4x\left(\frac{d\widetilde\sg}{dx}\right)^3+(4k^2+x^2+4\widetilde\sg)\left(\frac{d\widetilde\sg}{dx}\right)^2+x(N^2+2Nk-2\widetilde\sg)\frac{d\widetilde\sg}{dx}+
(\widetilde\sg-N(N+2k))\widetilde\sg
\end{equation}
with the initial data
\begin{equation}\label{zz9}
\widetilde\sg(0)=0,\quad \widetilde\sg'(0)=0\,.
\end{equation}
Equation \eqref{zz8} is degenerate at $x=0$ and the Cauchy-Kovalevskaya theorem is not applicable here.
In fact, the initial value problem \eqref{zz8}--\eqref{zz9} has a trivial solution, $\widetilde \sg=0$, and a nontrivial solution.
We are looking for a nontrivial solution to equation \eqref{zz8} as a power series,
\begin{equation}\label{zz10}
\widetilde\sg(x)=\sum_{j=2}^\infty \alpha_{j} x^{j}\,,
\end{equation}
and we find recursively that
\begin{equation}\label{zz11}
\alpha_2(16k^2\alpha_2+N^2+2Nk-4\alpha_2)=0,\quad \alpha_2\not=0 \implies
\alpha_2=-\frac{N(N+2k)}{4(4k^2-1)}\,,
\end{equation}
\begin{equation}\label{zz12}
\frac{4\alpha_3 N(N+2k)(k^2-1)}{4k^2-1}=0,\quad k\not=1 \implies
\alpha_3=0\,,
\end{equation}
\begin{equation}\label{zz13}
\begin{aligned}
\alpha_4=\frac{N(N+2k)(2N+2k-1)(2N+2k+1)}{16(4k^2-1)^2(4k^2-9)},
\end{aligned}
\end{equation}
\begin{equation}\label{zz14}
\frac{4\alpha_5 N(N+2k)(k^2-4)}{4k^2-1}=0,\quad k\not=2 \implies
\alpha_5=0\,,
\end{equation}
\begin{equation}\label{zz15}
\begin{aligned}
\alpha_6=-\frac{N(N+2k)(2N+2k-1)(2N+2k+1)(6N^2+12Nk+4k^2-1)}{32(4k^2-1)^3(4k^2-9)(4k^2-25)},
\end{aligned}
\end{equation}
\begin{equation}\label{zz16}
\frac{4\alpha_7 N(N+2k)(k^2-9)}{4k^2-1}=0,\quad k\not=3 \implies
\alpha_7=0\,,
\end{equation}
and so on. An expression for $\alpha_8$ is too long to be presented here. 
Observe that 
the odd coefficients $\alpha_{2j+1}$ vanish as long as $k>j$. Indeed,
we have the equation
\begin{equation}\label{zz17}
\frac{4\alpha_{2j+1} N(N+2k)(k^2-j^2)}{4k^2-1}=0,
\end{equation}
hence
\begin{equation}\label{zz18}
 \alpha_{2j+1}=0,\quad j=1,2,\ldots, k-1,
\end{equation}
which is equivalent to \eqref{a13}. 
For $j=k$, equation \eqref{zz17} 
does not determine the coefficient $\alpha_{2j+1}=\alpha_{2k+1}$. This implies that the 
initial value problem \eqref{zz8}--\eqref{zz9} has a one-parameter family of solutions, corresponding
to different values of the coefficient $\alpha_{2k+1}$.
\end{proof}
\section{Calculating  the moments}
\label{CTM}
The goal of this section is to calculate the quantity  $F_N(h,k)$ defined in \eqref{FNint2}  for  $k$ and $h$ non-negative  integers.  
As  explained in the introduction, $F_N(h,k)$ is related to the 
generalised Laguerre polynomials by 
\begin{equation}\label{a31}
F_N(h,k)=\lim_{\ep \rightarrow 0} (-1)^\frac{k(k-1)}{2}2^{-2h}\int_{-\infty}^{\infty}
K_{2h} (\ep,y)e^{{-N|y|} }\det\left[L^{(2k-1)}_{N+k-1-(i+j)}(-2|y|)\right]_{i,j=0,\dots,k-1} dy,
\end{equation}
where $K_{2h} (\ep,y)$ is defined in \eqref{jc4}.
It is convenient to rewrite the latter formula as
\begin{equation}\label{a4}
F_N(h,k)=(-1)^\frac{k(k-1)}{2}\lim_{\ep \rightarrow 0} 2^{-2h}\int_{0}^{\infty}
K_{2h} (\ep,\frac{x}{2})f_k(x) dx,
\end{equation}
where
\begin{equation}\label{a5}
f_k(x):=(-1)^\frac{k(k-1)}{2}e^{{-\frac{N}{2} x} }\det\left[L^{(2k-1)}_{N+k-1-(i+j)}(-x)\right]_{i,j=0,\dots,k-1} .
\end{equation}
\begin{lemma}\label{Lemma_F}
The following identity is satisfied:
\begin{equation}\label{a15}
F_N(h,k)
=(-1)^{h}\,f_k^{(2h)}(0),\quad h\in \mathbb{N}_0,
\end{equation}
where $f_k(x)$ is given in \eqref{a5}.
\end{lemma}
\begin{proof}
By \eqref{a4},
\begin{equation}\label{a7}
\begin{aligned}
F_N(h,k)
&=\lim_{\ep \rightarrow 0} 2^{-2h}\frac{\partial^{2h}}{\partial\ep^{2h}}\,\int_{0}^{\infty}
\frac{\ep f_k(x) dx}{\pi(\ep^2+(\frac{x}{2})^2)}\\
&=2\lim_{\ep \rightarrow 0} \frac{\partial^{2h}}{\partial\ep^{2h}}\,\int_{0}^{\infty}
\frac{ f_k(\ep y) dy}{\pi(1+y^2)}\\
&=2\lim_{\ep \rightarrow 0} \,\int_{0}^{\infty}
\frac{y^{2h} f_k^{(2h)}(\ep y) dy}{\pi(1+y^2)}\,.
\end{aligned}
\end{equation}
We have that
\begin{equation}\label{a8}
y^{2h}=(y^2+1)\left(y^{2h-2}-y^{2h-4}+\ldots+(-1)^{h+1}\right)+(-1)^{h},
\end{equation}
hence
\begin{equation}\label{a9}
\int_{0}^{\infty}
\frac{y^{2h} f_k^{(2h)}(\ep y) dy}{1+y^2}
=\int_{0}^{\infty}\left(y^{2h-2}-y^{2h-4}+\ldots+(-1)^{h+1}+\frac{(-1)^{h}}{y^2+1}\right)
 f_k^{(2h)}(\ep y) dy.
\end{equation}
Integrating by parts $2j$ times, we obtain that
\begin{equation}\label{a10}
\begin{aligned}
\int_{0}^{\infty}y^{2j} f_k^{(2h)}(\ep y) dy
&=\ep^{-2j}(2j)!\int_{0}^{\infty} f_k^{(2h-2j)}(\ep y) dy\\
&= -\ep^{-2j-1}(2j)!f_k^{(2h-2j-1)}(0),\quad 0\le j\le h-1,
\end{aligned}
\end{equation}
hence
\begin{equation}\label{a12}
\lim_{\ep \rightarrow 0}\int_{0}^{\infty}
\frac{y^{2h} f^{(2h)}(\ep y) dy}{1+y^2}
=(-1)^h\frac{\pi}{2}\,f_k^{(2h)}(0)+\lim_{\ep \rightarrow 0}\sum_{j=0}^{h-1}(-1)^{h-j} \ep^{-2j-1}(2j)!f_k^{(2h-2j-1)}(0).
\end{equation}
Since $F_N(h,k)$ is finite, all terms in the latter sum vanish, so that
\begin{equation}\label{a13}
f_k^{(2h-2j-1)}(0)=0,\quad 0\le j\le h-1.
\end{equation}
Thus,
\begin{equation}\label{a14}
\lim_{\ep \rightarrow 0}\int_{0}^{\infty}
\frac{y^{2h} f^{(2h)}(\ep y) dy}{1+y^2}
=(-1)^h\frac{\pi}{2}\,f_k^{(2h)}(0).
\end{equation}
Therefore we get the statement  of the lemma.
\end{proof}

\subsection{Evaluation of \texorpdfstring{$F_N(0,k)$}{TEXT}}

From \eqref{a14} with $h=0$ we have that
\begin{equation}\label{z1}
\begin{aligned}
F_N(0,k)&=(-1)^\frac{k(k-1)}{2} \det\left[L^{(2k-1)}_{N+k-1-(i+j)}(0)\right]_{i,j=0,\dots,k-1}\\
&=(-1)^\frac{k(k-1)}{2} \det\left[\frac{(N+3k-2-i-j)!}{(2k-1)!(N+k-1-i-j)!}\right]_{i,j=0,\dots,k-1}\\
&=\frac{(-1)^\frac{k(k-1)}{2}}{[(2k-1)!]^k} \det\left[\frac{(N+3k-2-i-j)!}{(N+k-1-i-j)!}\right]_{i,j=0,\dots,k-1}\,.
\end{aligned}
\end{equation}
In particular,
\begin{equation}\label{z2}
\begin{aligned}
F_N(0,1)=N+1,
\end{aligned}
\end{equation}
\begin{equation}\label{z3}
\begin{aligned}
F_N(0,2)&=-\frac{1}{36} 
\det\left[
\begin{matrix}
(N+2)(N+3)(N+4) & (N+1)(N+2)(N+3) \\
(N+1)(N+2)(N+3) & N(N+1)(N+2)
\end{matrix}\right]\\
&=-\frac{1}{36} \,(N+1)(N+2)^2(N+3)
\det\left[
\begin{matrix}
N+4 & N+1 \\
N+3 & N
\end{matrix}\right]\\
&=\frac{1}{12}\,(N+1)(N+2)^2(N+3)\,,
\end{aligned}
\end{equation}
\begin{equation}\label{z4}
\begin{aligned}
F_N(0,3)=&-\frac{1}{(5!)^3} \,(N+1)(N+2)^2(N+3)^3(N+4)^2(N+5)\\
&\times\det\left[
\begin{matrix}
(N+6)(N+7) & (N+2)(N+6) & (N+1)(N+2) \\
(N+5)(N+6) & (N+1)(N+5) & N(N+1) \\
(N+4)(N+5) & N(N+4) & (N-1)N
\end{matrix}\right]\\
=&\frac{1}{3\cdot 4!\cdot 5!}\,(N+1)(N+2)^2(N+3)^3(N+4)^2(N+5)\,,
\end{aligned}
\end{equation}
and so on. In general,
\begin{equation}\label{z5}
\begin{aligned}
F_N(0,k)
=&C_k(N+1)(N+2)^2\cdots (N+k-1)^{k-1}(N+k)^k\\
&\times(N+k+1)^{k-1}\cdots (N+2k-2)^2(N+2k-1)\,.
\end{aligned}
\end{equation}
To find the constant $C_k$, consider $N=0$:
\begin{equation}\label{z6}
\begin{aligned}
F_0(0,k)
=C_k\cdot 1\cdot 2^2 \cdots (k-1)^{k-1}k^k(k+1)^{k-1}\cdots (2k-2)^2(2k-1)\,.
\end{aligned}
\end{equation}
On the other hand, by \eqref{z1},
\begin{equation}\label{z7}
F_0(0,k)=\frac{(-1)^\frac{k(k-1)}{2}}{[(2k-1)!]^k} \det\left[\frac{(3k-2-i-j)!}{(k-1-i-j)!}\right]_{i,j=0,\dots,k-1}=1\,,
\end{equation}
hence
\begin{equation}\label{z8}
C_k=\frac{1}{1\cdot 2^2 \cdots (k-1)^{k-1}k^k
(k+1)^{k-1}\cdots (2k-2)^2(2k-1)}\,.
\end{equation}
Thus, we have the formula  in \cite{KeatingS:2000},
\beq
\begin{aligned}
\label{z9}
F_N(0,k)=&\dfrac{G(N+2k+1)G(N+1)G(k+1)^2}{G(N+k+1)^2G(2k+1)}\\
=&\frac{1}{1\cdot 2^2 \cdots (k-1)^{k-1}k^k
(k+1)^{k-1}\cdots (2k-2)^2(2k-1)}\\
&\times(N+1)(N+2)^2\cdots (N+k-1)^{k-1}(N+k)^k\\
&\times(N+k+1)^{k-1}\cdots (N+2k-2)^2(N+2k-1)\,,
\end{aligned}
\eeq
where $G(z)$ is the Barnes $G$-function (see Appendix~\ref{Barnes}).
In particular, this implies that
\begin{equation}\label{z10}
F(0,k)=\lim_{N\to\infty} \frac{F_N(0,k)}{N^{k^2}}=\dfrac{G(k+1)^2}{G(2k+1)},
\end{equation}
as shown in \cite{KeatingS:2000}.

\subsection{Evaluation of \texorpdfstring{$F_N(h,k)$}{TEXT}}

In order to evaluate the function $F_N(h,k)$ we use the following identities between
the Hankel determinant $H_k$ in \eqref{Sum3} and the function $f_k(x)$ defined in \eqref{a5},
\begin{equation}\label{zz3}
H_k(x)=(2\pi i)^k e^{kx} \det\left[L^{(2k-1)}_{N+k-1-(i+j)}(- x)\right]_{i,j=0,\dots,k-1}=(2\pi i)^k e^{(k+\frac{N}{2})x}
(-1)^{\frac{k(k-1)}{2}+h}\, f_k(x).
\eeq
Furthermore, the function $H_k(x)$ is related to  a  solution of the $\sg$-Painlev\'e V equation \eqref{painleveV}
\begin{equation}\label{zz6}
\dfrac{d}{dx}\log H_k=\frac{\sg_k(x)+kx+Nk}{x}\,,
\end{equation}
so that combining the above two relations we obtain
\beq
\label{f}
x\dfrac{d}{dx}\log f_k(x)=\sg_k(x)-\frac{N}{2}x+Nk=\ts(x)=\sum_{j=2}^{\infty}\alpha_jx^j\,,
\eeq
where $\ts(x)$ has been defined in \eqref{zz7} and the first few coefficient $\alpha_j$ have been evaluated in \eqref{zz11}--\eqref{zz16}.
Integrating the above equation, we obtain that
\begin{equation}\label{zz22}
\begin{aligned}
f_k(x)=&f_k(0)
\exp\left(\sum_{j=2}^\infty \frac{\alpha_{j} x^{j}}{j}\right)=f_k(0) \left[\sum_{j=0}^{k}\beta_{2j}x^{2j}+\sum_{j=2k+1}^{\infty}\beta_jx^j\right]\\
=&\left[1+\frac{1}{2}\,\alpha_2x^2+\left(\frac{\alpha_4}{4}+\frac{\alpha_2^2}{8}\right)x^4+\left(\frac{\alpha_6}{6}+\frac{\alpha_2\alpha_4}{8}+\frac{\alpha_2^3}{48}\right)x^6\right.\\
&\left.+
\left(\a_8+\a_2\a_6+\frac{1}{2}\a_4\a_4^2+\frac{1}{2}\a_4^2+\frac{1}{24}\a_2^4\right)x^8+\cdots \right]\,
\end{aligned}
\end{equation}
where  the coefficients $\beta_j$ are obtained from the above expansion  and using the explicit expression of $\alpha_j$  in \eqref{zz11}--\eqref{zz16} as 
\begin{equation}
\label{beta}
\beta_2=\frac{1}{2}\,\alpha_2= \frac{N(N+2k)}{4(4k^2-1)},\quad \beta_4=\frac{\alpha_4}{4}+\frac{\alpha_2^2}{8}=\frac{N(N+2k)(N^2+2kN+2)}{128(4k^2-1)(4k^2-9)},\quad \dots
\end{equation}
and so on.
Observe that by \eqref{zz16}, all odd powers $x^{2j+1}$ will be missing on the right-hand side as long as $j<k$. In other words,
\begin{equation}\label{zz23}
\begin{aligned}
f_k^{2j+1}(0)=0,\quad j=0,1,\ldots, k-1,
\end{aligned}
\end{equation}
which is equivalent to equation \eqref{a13}.
From Lemma~\ref{Lemma_F} we have the relation
\begin{equation}
\label{zz24}
F_N(h,k)
=(-1)^{h}\,f_k^{(2h)}(0)
\end{equation}
which implies, using the explicit expressions of $\alpha_j$   in \eqref{zz11}--\eqref{zz16},
\begin{equation}
\label{zz25}
\begin{aligned}
F_N(1,k)=& \frac{N(N+2k)}{4(4k^2-1)}\,F_N(0,k)\,,\\
F_N(2,k)=& \frac{3N(N+2k)(N^2+2kN+2)}{16(4k^2-1)(4k^2-9)}\,F_N(0,k),\\
F_N(3,k)=&-\dfrac{15}{64}\frac{N (2 k + N) }{(4 k^2-25) (k^2-9) (4k^2-1)^2}(-16 + 64 k^2 + 20 k N + 48 k^3 N + 10 N^2  \\
& \qquad \qquad -12 k^2 N^2 + 16 k^4 N^2 - 36 k N^3 + 16 k^3 N^3 - 9 N^4 + 4 k^2 N^4)\,F_N(0,k),
\end{aligned}
\end{equation}
and so on. A formula for $F_N(4,k)$ is too long to be presented here. 
\subsection{Scaling limit of Painlev\'e V as \texorpdfstring{$N\to\infty$}{TEXT} and the second main result}

From equations \eqref{zz11}, \eqref{zz13}, and \eqref{zz15} we find  
\begin{equation}\label{sl1}
\begin{aligned}
\xi_2&:=\frac{\alpha_2}{N^2}=-\frac{N(N+2k)}{4(4k^2-1)N^2}\underset{N\to\infty}{\longrightarrow} -\frac{1}{4(4k^2-1)}\,,\\
\xi_4&:=\frac{\alpha_4}{N^4}\underset{N\to\infty}{\longrightarrow}\frac{1}{4(4k^2-1)^2(4k^2-9)}\,,\\
\xi_6&:=\frac{\alpha_6}{N^6}\underset{N\to\infty}{\longrightarrow} -\frac{3}{4(4k^2-1)^3(4k^2-9)(4k^2-25)}\,,\\
\end{aligned}
\end{equation}
and so forth. Define, therefore, the function
\begin{equation}\label{sl2}
\begin{aligned}
\xi(t)=\sum_{j=2}^\infty \xi_j t^j:=\sum_{j=0}^\infty \frac{\alpha_j t^j}{N^j}= \widetilde{\sg}\left(\frac{t}{N}\right).
\end{aligned}
\end{equation}
From equation \eqref{zz8} we obtain that $\xi(t)$ satisfies the equation
\begin{equation}\label{sl3}
\begin{aligned}
\left(t\frac{d^2\xi}{dt^2}\right)^2= & - 4t\left(\frac{d\xi}{dt}\right)^3+\left(4k^2+4\xi+N^{-2}t^2\right)\left(\frac{d\xi}{dt}\right)^2+t(1+2N^{-1}k-2N^{-2}\xi)\frac{d\xi}{dt}\\
&-(1+2N^{-1}k-N^{-2}\xi)\xi.
\end{aligned}
\end{equation}
In the limit $N\to\infty$ it reduces to the  equation
\begin{equation}\label{sl4}
\begin{aligned}
\left(t\frac{d^2\xi}{dt^2}\right)^2&=- 4t\left(\frac{d\xi}{dt}\right)^3+\left(4k^2+4\xi\right)\left(\frac{d\xi}{dt}\right)^2 + t\frac{d\xi}{dt}-\xi.
\end{aligned}
\end{equation}
This equation can be identified as a special case of the 
$\sg$-Painlev\'e III$^\prime$ equation 
\eq
\label{sigma-Painleve-III}
\left(t\frac{d^2v}{dt^2}\right)^2 + \theta_0\theta_\infty\frac{dv}{dt} - \left(4\left(\frac{dv}{dt}\right)^2-1\right)\left(v-t\frac{dv}{dt}\right)-\frac{1}{4}(\theta_0^2+\theta_\infty^2)=0
\endeq
considered by Okamoto \cite[Proposition 1.7]{Okamoto:1987b} by choosing 
$\xi=v-k^2$, $\theta_0=0$, and $\theta_\infty=-2k$.
The initial data for \eqref{sl4} are
\begin{equation}\label{sl5}
\begin{aligned}
\xi(0)=0,\quad \xi'(0)=0.
\end{aligned}
\end{equation}
As before, the condition on $\xi'(0)$ follows from $\xi(0)=0$ and equation \eqref{sl3}.
We are looking for a solution 
\begin{equation}\label{sl6}
\xi(t)=\sum_{j=2}^\infty \xi_jt^j
\end{equation}
 to equation  \eqref{sl4} 
such that 
\begin{equation}\label{sl7}
\xi_2\not=0.
\end{equation}
We have that
\begin{equation}\label{sl8}
 \xi_{2j+1}=0, \quad j=1,2,\ldots,k-1. 
\end{equation}
The even coefficients, $\xi_{2j}$,
$j=1,2,\ldots,k$, can be found recursively from equation \eqref{sl4}.
From equations \eqref{z10}, \eqref{zz22}, and  \eqref{zz24}  we can consider the rescaled quantity
\beq
\frac{F_N(h,k)}{N^{k^2+2h}}=(-1)^{h}\left.\frac{d^{2h}}{dt^{2h}}f_k\left(\frac{t}{N}\right)\right|_{t=0}.
\eeq
We arrive at our second main result. 
\begin{theorem}
\label{Theorem2}
 The limit
\begin{equation}\label{sl10}
F(h,k)=\lim_{N\to\infty} \frac{F_N(h,k)}{N^{k^2+2h}}
\end{equation}
exists, and it is given by the formula 
\begin{equation}\label{sl9}
F(h,k)
=(-1)^{h}\,F(0,k)\frac{d^{2h}}{dt^{2h}}
\left[\exp\left(\sum_{j=2}^\infty \frac{\xi_{j} t^{j}}{j}\right)\right]\Bigg|_{t=0},
\end{equation}
 where $F(0,k)$ is given in \eqref{z10} and   the power
series \eqref{sl6}  solves equation \eqref{sl4}.
\end{theorem}

It is noteworthy that equation \eqref{sl4} is closely related to the 
$\sg$-Painlev\'e III (as opposed to III$^\prime$) equation.
Namely, let
\begin{equation}\label{sl11}
\sigma_{III}(s):=2\xi(s^2)+k^2\,.
\end{equation}
Then $\sigma_{III}(s)$ solves the $\sg$-Painlev\'e III equation (cf. 
\cite[(C.29)]{JM}),
\beq
\begin{split}
\left(s\frac{d^2\sigma_{III}}{ds^2} - \frac{d\sigma_{III}}{ds}\right)^2
 = & 4\left(2\sigma_{III} - s \frac{d\sigma_{III}}{ds}\right)\left(\left(\frac{d\sigma_{III}}{ds}\right)^2 -4s^2\right) \\
  & +2(\theta_0^2 + \theta_{\infty}^2)\left(\left(\frac{d\sigma_{III}}{ds}\right)^2 +4s^2\right) - 16\theta_0\theta_{\infty}s\frac{d\sigma_{III}}{ds},
\end{split}
\eeq
for $\theta_0=0$ and $\theta_{\infty}=-2k$
with the initial conditions
\begin{equation}\label{sl13}
\sigma_{III}(0)=k^2\,,\quad \sigma_{III}'(0)=0.
\end{equation}

We write the first few values of $F(h,k)$:
\begin{equation}\label{zz26}
\begin{aligned}
F(1,k)&= \frac{F(0,k)}{4(4k^2-1)}\,,\\
F(2,k)&= \frac{3F(0,k)}{16(4k^2-1)(4k^2-9)}\,,\\
F(3,k)&= \frac{15F(0,k)}{64(4k^2-1)^2(4k^2-25)}\,,\\
F(4,k)&= \frac{105(4k^2-33)F(0,k)}{256(4k^2-1)^2(4k^2-9)(4k^2-25)(4k^2-49)}\,,\\
F(5,k)&= \frac{925(16k^4-360k^2+1497)F(0,k)}{1024(4k^2-1)^2(4k^2-9)^2(4k^2-25)(4k^2-49)(4k^2-81)}\,,
\end{aligned}
\end{equation}
and so on.
In \cite{Dehaye:2008}  an explicit  formula for the above expansions has been derived and it takes the form
\beq
\label{FD}
F(h,k)=\dfrac{(2h)!}{h!2^{3h}}F(0,k)\dfrac{\widetilde{X}_{2h}(2k)}{Y_{2h}(2k)},\quad 
\eeq
where the polynomials $\widetilde{X}_{2h}(s)$ are obtained in a combinatorial way and we report the first   few  (see Table 4 in \cite{Dehaye:2008}):
\beq
\widetilde{X}_{2}(s)=1,\quad \widetilde{X}_{4}(s)=1,\;\;\widetilde{X}_{6}(s)=s^2-9,\;\;\widetilde{X}_{8}(s)=s^2-33,\;\;\widetilde{X}_{10}(s)=s^4-90s^2+1497,
\eeq
and the polynomial $Y_{2h}(s)$ is given by the expression
\beq
Y_r(s)=\prod_{\stackrel{1\leq a\leq r-1}{a \;\text{odd}}}(s^2-a^2)^{\eta_a(r)},\quad \eta_a(r)=\floor*{ \dfrac{-a+\sqrt{a^2+4r}}{2}},
\eeq
where the symbol $\floor*{\,z\,} $ denotes the integer part  of $z$.
We have
\beq
\begin{split}
Y_2(s)=s^2-1,\;\;Y_4(s) & =(s^2-9)(s^2-1),\;\;Y_6(s)=(s^2-1)^2(s^2-9)(s^2-25),\;\;\\
Y_8(s) & =(s^2-1)^2(s^2-9)(s^2-25)(s^2-49), \\
Y_{10}(s)  =(&s^2-1)^2(s^2-9)^2(s^2-25)(s^2-49)(s^2-81).
\end{split}
\eeq
Combining the above expressions we  can verify that the formula \eqref{FD} reproduces the terms obtained in \eqref{zz26}.
\section{Scaling Limit of the  Riemann-Hilbert Problem as \texorpdfstring{$N \to \infty$}{TEXT}}

In this section we will supplement  the result of the previous section by performing the large-$N$ scaling limit
directly in the $X$-Riemann-Hilbert problem (\ref{PV13})--(\ref{PV15}).  We start by changing the jump contour of the 
problem.

Let $\Gamma_N$ be the positively oriented  circle of radius $N$ centered  at $z =0$ and let us pass from  the original function $X\equiv X_{N}(z,x; k,n))$ to the
new  matrix valued function $\widetilde{X} \equiv \widetilde{X}_{N}(z, x; k,n)$ according to the following  rule.
\begin{itemize} 
\item For all $z$ inside the small circle $\Gamma$  and outside the big circle $\Gamma_N$, we put
\beq
\widetilde{X}_N (z) := X_N(z).
\eeq
\item For all $z$ between the circles we define
\beq
\widetilde{X}_N (z):= X_N(z)\begin{pmatrix}
1 & e^{xz}\frac{z^{3k+N}}{(z-1)^{N+k}} \\
0 & 1
\end{pmatrix}.
\eeq
\end{itemize}
The new Riemann-Hilbert problem reads
\beq\label{PV13100}
\widetilde{X}_N(z) \in \mathcal{H} ({\mathbb C}P^1 \setminus (\Gamma_N\cup\{0\} )),
\eeq
\beq\label{PV14100}
\widetilde{X}_{N +}(z)=\widetilde{X}_{N -}(z) \begin{pmatrix}
1 & e^{xz}\frac{z^{3k+N}}{(z-1)^{N+k}} \\
0 & 1
\end{pmatrix},
\quad z \in \Gamma_N,
\eeq
\beq\label{PV15100}
\widetilde{X}_N(z) \sim \left(I + \mathcal{O}(z) + \cdots \right )(-z)^{-n\sigma_3},  \quad z \rightarrow 0.
\eeq
The $\widetilde{X}$-Riemann-Hilbert problem is ready for the large-$N$ 
scaling limit.
Put
\beq
Z_N(z, t; k,n):=  \widetilde{X}_N\left(zN, \frac{t}{N}; k,n\right)N^{n\sigma_3}.
\eeq
Then, the Riemann-Hilbert problem (\ref{PV13100})--(\ref{PV15100}) 
transforms to the following Riemann-Hilbert problem for the function $Z_N$ 
posed on the unit circle $\Gamma_1 := \{z: |z| = 1\}$:
\beq\label{PV13100b}
Z_N(z) \in \mathcal{H} ({\mathbb C}P^1 \setminus (\Gamma_1\cup\{0\} )),
\eeq
\beq\label{PV14100b}
Z_{N +}(z)=Z_{N -}(z) \begin{pmatrix}
1 & e^{tz}\frac{z^{2k} N^{2(k-n)}}{\left(1-\frac{1}{zN}\right)^{N+k}} \\
0 & 1
\end{pmatrix},
\quad z \in \Gamma_1,
\eeq
\beq\label{PV15100b}
Z_N(z) \sim \left(I + \mathcal{O}(z) + \cdots \right )(-z)^{-n\sigma_3},  \quad z \rightarrow 0.
\eeq
Assume now, and this is the case of our main concern, that 
$$
n=k.
$$
Then, as $N \to \infty$ the  new jump matrix converges to the matrix
$$
 \begin{pmatrix}
1 & z^{2k}e^{tz + \frac{1}{z}} \\
0 & 1
\end{pmatrix}
$$
uniformly for $z\in \Gamma_1$ and $t\in K$, where $K$ is a compact set in $\mathbb C$.
By standard Riemann-Hilbert arguments, this implies the convergence
of the function $Z_{N}(z, t)$ to the function $Z(z,t)$ satisfying 
Riemann-Hilbert problem
\beq\label{PIII1}
Z(z) \in \mathcal{H} ({\mathbb C}P^1 \setminus (\Gamma_1\cup\{0\} )),
\eeq
\beq\label{PIII2}
Z_{+}(z)=Z_ {-}(z)
 \begin{pmatrix}
1 & z^{2k}e^{tz + \frac{1}{z}} \\
0 & 1
\end{pmatrix}, 
\quad z \in \Gamma_1,
\eeq
\beq\label{PIII3}
Z(z) = \left(I + \mathcal{O}(z) \right )(-z)^{-k\sigma_3},  \quad z \rightarrow 0.
\eeq
In fact, the estimate 
\begin{equation}\label{ZZN}
Z_N(z,t) =\left(I + \mathcal{O}\left(\frac{1}{N(1 +|z|)}\right)\right)Z(z,t), \quad N \to \infty,
\end{equation}
holds uniformly for all $z \in {\mathbb C}P^1 $ and $t\in K$, where $K$ is a compact set in $\mathbb C$.

Denote by $\phi^{Z_N}_1\equiv \phi^{Z_N}_1(t)$ and  $\phi^{Z}_1\equiv \phi^{Z}_1(t)$ the first matrix coefficients in the expansions 
near $z=\infty$ of the functions $Z_{N}(z)$ and  $Z(z)$, respectively:
\begin{equation}
Z_{N}(z) = Z_N(\infty)\left(I + \frac{\phi^{Z_N}_1}{z} + ....\right), \quad z\to \infty
\end{equation}
and 
\begin{equation}
Z(z) = Z(\infty)\left(I + \frac{\phi^{Z}_1}{z} + ....\right), \quad z\to \infty.
\end{equation}
Estimate (\ref{ZZN}) implies that
\begin{equation}\label{phi1ZNZ}
\phi^{Z_N}_1(t)\underset{N\to\infty}{\longrightarrow} \phi^{Z}_1(t)
\end{equation}
uniformly for all  $t\in K$, where $K$ is a compact set in $\mathbb C$. At the same time,  
recalling the connection of $Z_{N}(z)$ with the function $X(z)$ and the connection of the latter with $\Phi(z) \equiv \Phi_{N}(z,t;k)$,
we arrive at the relation
\begin{equation}\label{phi1Nphi} 
\phi_{1,11}\left(\frac{t}{N}\right) = N\phi^{Z_N}_{1,11}(t) + \frac{N+k}{2}.
\end{equation}
From this relation and (\ref{phi1ZNZ}) we have
\begin{equation}\label{phi1NphiZ} 
\frac{1}{N}\phi_{1,11}\left(\frac{t}{N}\right)\underset{N\to\infty}{\longrightarrow} \phi^{Z}_{1,11}(t) + \frac{1}{2}.
\end{equation}
Together with Lemma \ref{Hn-m-lemma}, the last limit allows us to find the large-$N$ scaling limit of the Hankel determinant $H_k[w_0] \equiv H(x)$ in
terms of the solution of the $Z$-Riemann-Hilbert problem: 
\begin{equation}\label{HphiZ}
\frac{d}{dt } \log H_k\left(\frac{t}{N}\right)\underset{N\to\infty}{\longrightarrow} -\phi^{Z}_{1,11}(t).
\end{equation}
In the previous subsection we have already connected  this limit with the special solution $\xi(z)$  of the third Painlev\'e equation
(\ref{sl4}). Let us show how the same result can be derived from the $Z$-Riemann-Hilbert problem 
(\ref{PIII1})--(\ref{PIII3}).

Similar to the finite-$N$ case, we introduce  the function (cf. (\ref{PV16})),
\begin{equation}\label{PhiIIIdef}
\Phi^{(III)}(z,s) := s^{k\sigma_3}Z^{-1}(\infty)Z\left(\frac{z}{s}, s^2\right)\left(\frac{z}{s}\right)^{k\sigma_3}
e^{\frac{s}{2}\left(z + \frac{1}{z}\right)\sigma_3}.
\end{equation}
The Riemann-Hilbert problem (\ref{PIII1})--(\ref{PIII3}) in terms of $\Phi^{(III)}(z,t)$ reads
\beq\label{PIII11}
\Phi^{(III)}(z) \in \mathcal{H} ({\mathbb C} \setminus (\Gamma_1\cup\{0\} )),
\eeq
\beq\label{PIII21}
\Phi^{(III)}_{+}(z)=\Phi^{(III)}_ {-}(z)
 \begin{pmatrix}
1 & 1 \\
0 & 1
\end{pmatrix},
\quad z \in \Gamma_1,
\eeq
\beq\label{PIII33}
\Phi^{(III)}(z) = P_0\Bigl(I + \mathcal{O}(z) \Bigr )e^{\frac{s}{2z}\sigma_3},  \quad z \rightarrow 0,
\eeq
\beq\label{PIII34}
\Phi^{(III)}(z) = \left(I + \frac{\phi^{(III)}_1}{z}  + ....\right )z^{k\sigma_3}e^{\frac{sz}{2}\sigma_3},  \quad z \rightarrow \infty,
\eeq
where 
\begin{equation}\label{phi1III}
\phi^{(III)}_1(s) = ss^{k\sigma_3}\phi^{Z}_{1}(s^2)s^{-k\sigma_3} + \frac{s}{2}\sigma_3
\end{equation}
and 
\begin{equation}\label{P0}
P_0 := (-s)^{k\sigma_3}Z^{-1}(\infty).
\end{equation}
Repeating now the same standard argument based on Liouville's theorem as in \S3 which led us to the Painlev\'e V  Lax pair 
(\ref{PV23})--(\ref{PV24}),  we arrive at the following Lax pair for the function $\Phi^{(III)}(z,s)$:
\begin{equation}
\label{Lax100}
\dfrac{\partial\Phi^{(III)}}{\partial z}=\left(\dfrac{s}{2}\sigma_3+\dfrac{A_0}{z}+\dfrac{A_{-1}}{z^2}\right)\Phi^{(III)}(z,s), 
\end{equation}
\begin{equation}\label{Lax200}
\dfrac{\partial\Phi^{(III)}}{\partial s}=\left(\dfrac{z}{2}\sigma_3+B_0+ \dfrac{B_{-1}}{z}\right)\Phi^{(III)}(z,s),
\end{equation}
which is {\it exactly} the Jimbo-Miwa Lax pair for the Painlev\'e III equation (see (C.18), (C.19) of  \cite{JM}). 
A comparison of the asymptotics (\ref{PIII33}) and (\ref{PIII34}) with the  formulae (C.26) and (C.25)
of \cite{JM}, respectively, shows that the Jimbo-Miwa formal monodromy parameters $\theta_{\infty}$ and
$\theta_0$ are
\begin{equation}\label{thetas}
\theta_{\infty} = -2k, \quad \theta_0 =0.
\end{equation}
Following \cite{JM}, we parameterise the matrix coefficients $A_0, A_{-1}$ and $B_0, B_{-1}$ by  
the functional parameters $y, w, u, v$  according to the equations ($z^{JM} \equiv w$, $t^{JM} = s$)
\beq\label{5P21}
A_0 := \begin{pmatrix}
k& u \\
v & -k
\end{pmatrix} = sB_0 + k\sigma_3,
\eeq
\beq\label{5P23}
A_1=\begin{pmatrix}
w-\frac{s}{2} & \frac{u}{y}\\
-\frac{yw}{u}(w-s)& -w +\frac{s}{2}
\end{pmatrix} = -sB_{-1}.
\eeq
We will also need the following formula for the matrix coefficient $\phi^{(III)}_1$ which can be obtained 
by  the substitution of the expansion (\ref{PIII34}) into the equation (\ref{Lax100}):
\begin{equation}\label{phiIII1uv}
\phi^{(III)}_1 = \begin{pmatrix}
-s^{-1}uv - w +\frac{s}{2}& -s^{-1}u\\
s^{-1}v &s^{-1}uv + w -\frac{s}{2}
\end{pmatrix}.
\end{equation}
The compatibility condition of equations (\ref{Lax100}) and (\ref{Lax200}) yields the following deformation equations on
$y(s)$ and $w(s)$ (cf. \cite[(C.23)]{JM}):
\begin{equation}\label{defyP3}
s\frac{dy}{ds} = 4wy^2 -2sy^2 + (2\theta_{\infty} -1)y + 2s
\end{equation} 
and
\begin{equation}\label{defwP3}
s\frac{dw}{ds} = -4yw^2 -2sy^2 + (4sy -2\theta_{\infty} +1)w + (\theta_0 + \theta_{\infty})s.
\end{equation} 
We recall that in our case $\theta_0 = 0$ and $\theta_{\infty} = -2k$. In turn, this system implies a single
second-order ODE, the Painlev\'e III (as opposed to III$^\prime$) equation 
for the function $y(s)$ (cf. \cite[(C.23)]{JM} or \cite{Okamoto:1987b}),
\begin{equation}\label{yP3}
\frac{d^2y}{ds^2}= \frac{1}{y}\left(\frac{dy}{ds}\right)^2 - \frac{1}{s}\frac{dy}{ds} +\frac{1}{s}(\alpha y^2 + \beta)
+ \gamma y^3 + \frac{\delta}{y},
\end{equation}
with
$$
\alpha = 4\theta_0 = 0, \quad \beta = 4(1-\theta_{\infty}) = 4(1+2k), \quad \gamma = -\delta = 4.
$$
Moreover,  from \cite{JM} one can extract the following addition expression  for the 11-entry of the matrix coefficient 
$\phi^{(III)}_1$ in the expansion (\ref{PIII34}) of the solution $\Phi^{(III)}(z)$ at $z = \infty$ (cf. (\ref{HamP5})):
\begin{equation}\label{phi1III11}
(\phi^{(III)}_1)_{11} = -\frac{1}{2}\mathscr{H}_{III} + \frac{k^2}{2t},
\end{equation}
where
\begin{equation}\label{P3Ham}
\mathscr{H}_{III} \equiv \mathscr{H}_{III}(y, w;s) = \frac{1}{s}\Bigl( 2w^2y^2 + 2w(s-sy^2-2ky) +2ksy - s^2 +k^2\Bigr)
\end{equation}
is the Hamiltonian of the  dynamical system (\ref{defyP3})--(\ref{defwP3}). Correspondingly, the tau function $\tau_{III}(s)$  and the sigma function
$\sigma_{III}(s)$ are defined by the equations
\begin{equation}\label{tauP3}
\frac{d\log\tau_{III}(s)}{ds} = \mathscr{H}_{III}(y(s), w(s); s) = \frac{\sigma_{III}(s)}{s},
\end{equation}
and the sigma-form of the third Painlev\'e equation (\ref{yP3}) reads (cf. \cite[(C.29)]{JM})
\beq
\begin{split}
\left(s\frac{d^2\sigma_{III}}{ds^2} - \frac{d\sigma_{III}}{ds}\right)^2
 = & 4\left(2\sigma_{III} - s \frac{d\sigma_{III}}{ds}\right)\left(\left(\frac{d\sigma_{III}}{ds}\right)^2 -4s^2\right)\\ 
  & +2(\theta_0^2 + \theta_{\infty}^2)\left(\left(\frac{d\sigma_{III}}{ds}\right)^2 +4s^2\right) - 16\theta_0\theta_{\infty}
s\frac{d\sigma_{III}}{dt}
\end{split}
\eeq
with
\beq
\theta_0 = 0, \quad \theta_{\infty} = -2k.
\eeq
Equations (\ref{tauP3}), (\ref{phi1III11}), and (\ref{phi1III}) yield the following formula for the matrix entry
$(\phi^{Z}_{1})_{11}$:
\begin{equation}\label{phiZ1sigma}
(\phi^{Z}_{1}(t))_{11} = \frac{k^2 - \sigma_{III}(\sqrt{t})}{2t}  - \frac{1}{2}.
\end{equation}
Noticing that 
$$
\sigma_{III}(\sqrt{t}) = 2\xi(t) + k^2,
$$
where $\xi(t)$ is the solution of the $\xi$-form of the third Painlev\'e equation (\ref{sl4}), we conclude that
\begin{equation}\label{phi1IIIxi}
(\phi^{Z}_{1}(t))_{11} = -\frac{\xi(t)}{t}  - \frac{1}{2}.
\end{equation}
This, together with (\ref{HphiZ}), yields the limit formula
\begin{equation}\label{Hkxi}
\frac{d}{dt } \log H_k\left(\frac{t}{N}\right)\underset{N\to\infty}{\longrightarrow}\frac{\xi(t)}{t}  + \frac{1}{2}.
\end{equation}
For our principal object, the  function $f_k(x)$ (see (\ref{zz3})), we have that
\begin{equation}\label{fkxi}
\frac{d}{dt } \log f_k\left(\frac{t}{N}\right)\underset{N\to\infty}{\longrightarrow}\frac{\xi(t)}{t}.
\end{equation}
This is our second main result, i.e. Theorem \ref{Theorem2}.

It is worth noticing that the importance of the $Z$-Riemann-Hilbert problem lies in the fact that it  can be used for the large-$t$ asymptotic analysis of the function $\xi(t)$ which will be needed in the case of the half integer $h$. Indeed, in that case the  expression 
for $F_N(h,k)$ in terms of $\sigma(s)$ and hence the expression for $F(h,k)$ in terms of $\xi(t)$ are not local -- see
the next section -- and therefore global information about the behavior of $\xi(t)$ is essential.

We conclude this section by making the following interesting observation concerning the  $Z$-Riemann-Hilbert
problem. Put
\begin{equation}\label{YIIIdef}
Y^{(III)}(z): = \Phi^{(III)}(z,s)e^{-\frac{s}{2}\left(z + \frac{1}{z}\right)\sigma_3}= s^{k\sigma_3}Z^{-1}(\infty)Z\left(\frac{z}{s}, s^2\right)\left(\frac{z}{s}\right)^{k\sigma_3}.
\end{equation}
Then, the Riemann-Hilbert problem (\ref{PIII11})--(\ref{PIII34}) in terms of $Y^{(III)}(z,t)$ reads
\beq\label{YIII11}
Y^{(III)}(z) \in \mathcal{H} ({\mathbb C} \setminus \Gamma_1),
\eeq
\beq\label{YIII21}
Y^{(III)}_{+}(z)=Y^{(III)}_ {-}(z)
 \begin{pmatrix}
1 & e^{s\left(z + \frac{1}{z}\right)} \\
0 & 1
\end{pmatrix},
\quad z \in \Gamma_1,
\eeq
\beq\label{YIII31}
Y^{(III)}(z) = \left(I + \frac{\phi^{YIII}_1}{z}  + ....\right )z^{k\sigma_3},  \quad z \rightarrow \infty,
\eeq
where 
\begin{equation}\label{phi1YIII}
\phi^{YIII}_1(s) = ss^{k\sigma_3}\phi^{Z}_{1}(s^2)s^{-k\sigma_3}.
\end{equation}
This is again an example of a Riemann-Hilbert  problem  from the theory of Hankel determinants. 
The corresponding contour and 
weight, this time, are the unit circle $\Gamma_1$ and the Bessel type weight
\begin{equation}\label{w0III}
w^{(III)}_0(z) = e^{s\left(z + \frac{1}{z}\right)},
\end{equation}
respectively. In fact, the Hankel determinant associated with the problem (\ref{YIII11})--(\ref{YIII31}) is the determinant
\begin{equation}\label{hankelB}
H^{(III)}_k(s) = \det\left[\int_{\Gamma_1}z^{i+j}e^{s\left(z+\frac{1}{z}\right)}dz\right]_{i,j = 0,...,k-1}
= (2\pi i)^k\det\left[I_{i+j +1}(2s)\right]_{i,j = 0,...,k-1}.
\end{equation}
The relations similar to (\ref{PV28}) take the form
\begin{equation}\label{PIII28}
\frac{H^{(III)}_{k+1}}{H^{(III)}_k}=h_k,\quad h_k=-2\pi i (\phi^{YIII}_1)_{12},\quad \frac{1}{h_{k-1}}=-\frac{1}{2\pi i}(\phi^{YIII}_1)_{21},
\end{equation}
where now $h_k := \int_{\Gamma_1}P^2_k(z)w^{(III)}_0(z)dz$ and $P_k(z)$  are monic polynomials orthogonal on $\Gamma_1$
with respect to the weight (\ref{w0III}). These relations, as in the case of our original Laguerre-Hankel determinant $H_n$, can be used to prove the following analogue of Lemma \ref{Hn-m-lemma}.
\begin{lemma}
\label{Hk0-m-lemma}
The following relation between $H^{(III)}_k(s)$ and $\phi^{YIII}_1$ 
holds:
\beq\label{Hk0m1tilde}
\frac{d}{ds}\log H_k^{(III)}(s) = -2\phi^{YIII}_{1,11}(s) +\frac{k^2}{s}.
\eeq
\end{lemma}
\begin{proof}
It is convenient to make yet another change-of-variable  transformation of the function $\Phi^{(III)}(z,s)$, namely,
\begin{equation}\label{Phitilde}
\Phi^{(III)}(z,s) \to  \widetilde{\Phi}^{(III)}(z,s) = s^{-k\sigma_3}\Phi^{(III)}(zs,s).
\end{equation}
The motivation for this transformation is that the $s$-equation of the Lax pair for the function $\widetilde\Phi^{(III)}(z,s)$ is considerably 
simpler. The new coefficient matrix is a linear function of $z$ while the coefficient matrix for $\Phi^{(III)}(z,s)$
also has a simple pole at $z = 0$;  see (\ref{Lax200}). Indeed, the $s$-equation for  $\widetilde{\Phi}^{(III)}(z,s)$
is a combination of both equations in the Lax pair (\ref{Lax100})--(\ref{Lax200}) for $\Phi^{(III)}(z,s)$
so that the pole at $z =0$ cancels out and we have
\begin{equation}\label{Laxtilde}
\dfrac{\partial\widetilde{\Phi}^{(III)}}{\partial s}=\Bigl(zs\sigma_3+\widetilde{B}_0\Bigr)\widetilde{\Phi}^{(III)}(z,s)\equiv \widetilde{B}(z)\widetilde{\Phi}^{(III)}(z,t),
\end{equation}
where
\beq
\widetilde{B}_0 = 2s^{-k\sigma_3}B_0s^{k\sigma_3} = 
\frac{1}{s}\begin{pmatrix}0&2us^{-2k}\\
2vs^{2k}&0\end{pmatrix}.
\eeq
One also can notice that the expansion (\ref{PIII34}) of the function $\Phi^{(III)}(z,s)$ at $z=\infty$ transforms
to the following expansion of the function $\widetilde{\Phi}^{(III)}(z,s)$ at $z=\infty$:
\beq\label{PIII34tilde}
\widetilde{\Phi}^{(III)}(z) = \left(I + \frac{\widetilde{\phi}^{(III)}_1}{z}  + ....\right )z^{k\sigma_3}e^{\frac{s^2z}{2}\sigma_3},  \quad z \rightarrow \infty,
\eeq
where 
\begin{equation}\label{phi1IIItilde}
\widetilde{\phi}^{(III)}_1(s) = \frac{1}{s}s^{-k\sigma_3}\phi^{(III)}_{1}(s)s^{k\sigma_3}=
\phi^{Z}_{1}(s^2) + \frac{1}{2}\sigma_3.
\end{equation}

The function $\widetilde{\Phi}^{(III)}(z,s)$ satisfies a differential equation with respect to $z$ as well, however we will not need it.
What we will need is the difference equation for $\widetilde{\Phi}^{(III)}(z,s)$ associated with the  shift $k \to k+1$. To derive
this equation we indicate explicitly the dependence of $\widetilde{\Phi}^{(III)}(z,s)$  on $k$,
\beq
\widetilde{\Phi}^{(III)}(z,s) \equiv \widetilde{\Phi}^{(III)}_k(z,s),
\eeq
and consider the discrete logarithmic derivative  $\widetilde{\Phi}^{(III)}_{k+1}\Bigl[\widetilde{\Phi}^{(III)}_{k}\Bigr]^{-1}(z)$. Since the jump matrix of the $\Phi^{(III)}$-Riemann-Hilbert problem does not depend on $k$ we, using again Liouville's theorem, will arrive
at the following difference equation:
\begin{equation}\label{Phitildedif}  
 \widetilde{\Phi}^{(III)}_{k+1}(z) =\Bigl( zU_0 + U_1^{(k)}\Bigr)\widetilde{\Phi}^{(III)}_{k}(z)\equiv U_k(z)\widetilde{\Phi}^{(III)}_{k}(z).
\end{equation}  
The matrix coefficients $U_0$ and $U_1$ can be determined via the substitution of the expansion (\ref{PIII34tilde}) into
(\ref{Phitildedif}). One finds 
\begin{equation}\label{U0}
U_0 = \begin{pmatrix} 1&0\\
0&0\end{pmatrix}
\end{equation}
and
\begin{equation}\label{U1}
U_1^{(k)} = \widetilde{\phi}^{(III,k+1)}_1\begin{pmatrix} 1&0\\
0&0\end{pmatrix} - \begin{pmatrix} 1&0\\
0&0\end{pmatrix}\widetilde{\phi}^{(III, k)}_1 \equiv
 \begin{pmatrix} a_{k+1} - a_k&-b_k\\
c_{k+1}&0\end{pmatrix},
\end{equation}
where $\widetilde{\phi}^{(III, k)}_1$ is the matrix coefficient from the expansion (\ref{PIII34tilde}) with the explicit
indication of the dependence on the integer $k$, and $a_k$, $b_k$, $c_k$  are temporary notations for 
the matrix entries of $\widetilde{\phi}^{(III)}_1$, i.e., 
\begin{equation}\label{phitilde1a}
a(s) := (\widetilde{\phi}^{(III)}_{1}(s))_{11} = (\phi^{Z}_{1}(s^2))_{11}  + \frac{1}{2},
\end{equation}
\begin{equation}\label{phitilde1b}
b(s) := (\widetilde{\phi}^{(III)}_{1}(s))_{12} = (\phi^{Z}_{1}(s^2))_{12} = -s^{-2-2k}u(s), 
\end{equation}
and
\begin{equation}\label{phitilde1c}
c(s) := (\widetilde{\phi}^{(III)}_{1}(s))_{21} = (\phi^{Z}_{1}(s^2))_{21} = -s^{-2+2k}v(s). 
\end{equation}
In presenting these formulae we have also taken into account equations (\ref{phiIII1uv}) and (\ref{phi1IIItilde})

The next step is to consider the compatibility condition of equations (\ref{Phitildedif}) and (\ref{Laxtilde}), that is, the
differential-difference equation
\beq
\frac{dU_{k}(s)}{ds} = \widetilde{B}_{k+1}(z) U_k(z) - U_k(z)\widetilde{B}_{k}(z).
\eeq
This equation, in particular, means that
\begin{equation}\label{difb}
-\frac{db_k}{ds} = 2s b_k (a_{k+1} - a_k).
\end{equation} 
From formulae (\ref{PIII28}), (\ref{phi1YIII}), and (\ref{phitilde1b}) it follows that
\beq
h_k(s) = -2\pi i (\phi^{YIII}_{1}(s))_{12} = -2\pi i s^{1+2k}(\phi^{Z}_{1}(s^2))_{12} =  -2\pi i s^{1+2k}b_k(s),
\eeq
and hence (\ref{difb}) becomes
\begin{equation}\label{difh}
\frac{d}{ds}\log h_k = \frac{1+2k}{s} -2s(a_{k+1} - a_k).
\end{equation}
By virtue of the first equation in (\ref{PIII28}), we immediately derive from (\ref{difh}) the differential identity 
\begin{equation}\label{difH000}
\frac{d}{ds}\log H^{(III)}_k(s) = -2sa_k + \frac{k^2-1}{s} + c(s)
\end{equation}
for the Hankel determinant $H^{(III)}_k$, where
\beq
c(s) = \frac{d}{ds}\log H^{(III)}_1 +2sa_1.
\eeq
Because of (\ref{phitilde1a}) and (\ref{phi1YIII}), to complete the proof of the Lemma we only need to show that
\begin{equation}\label{cs}
c(s) =  s + \frac{1}{s}.
\end{equation}
In order to see (\ref{cs}) we notice first that
\begin{equation}\label{HIII1}
H^{(III)}_1(s) = 2\pi i I_{1}(2s).
\end{equation}
Secondly, we use the fact that
\begin{equation}\label{a1alpha}
a_1 = s^{-1}c_0 + \frac{1}{2}, 
\end{equation}
where $c_0$ is the zero-degree coefficient in the orthogonal polynomial
\beq
P_1(z) = z + c_0, \quad \int_{\Gamma_1}P_1(z)w^{(III)}_0(z) dz =0,
\eeq
and hence
\begin{equation}\label{alpha000}
c_0 = -\frac{I_2(2s)}{I_1(2s)}.
\end{equation}
Equation (\ref{cs}) follows from (\ref{HIII1}), (\ref{a1alpha}), (\ref{alpha000}), and 
one of the classical differential identities for  the Bessel functions,
$$
\frac{dI_1(s)}{ds} = \frac{1}{s}I_1(s) + I_2(s).
$$
\end{proof}
Combining \eqref{phi1YIII}, (\ref{phi1IIIxi}), and  Lemma~\ref{Hk0-m-lemma}, we arrive at the following expression of $\frac{d}{dt}\log H^{(III)}_k(\sqrt{t}) $ in terms
of the Painlev\'e III function $\xi(t)$:
\beq\label{HIIIxi}
\frac{d}{dt}\log H^{(III)}_k(\sqrt{t})=\frac{k^2}{t} +\frac{1}{2} + \frac{\xi(t)}{t}.
\eeq
A comparison of this relation and the asymptotics  (\ref{Hkxi}) implies the following transition formula from 
the Laguerre-Hankel determinant $H_k$ to the Bessel-Hankel determinant $H^{(III)}$:
\begin{equation}\label{HkHIII}
\frac{d}{dt } \log H_k\left(\frac{t}{N}\right)\underset{N\to\infty}{\longrightarrow}
\frac{d}{dt}\log H^{(III)}_k(\sqrt{t}) -\frac{k^2}{t},
\end{equation}
or, more explicitly,
\begin{equation}\label{HkHIII2}
\frac{d}{dt } \log \det\left[L^{2k-1}_{N+k-1-(i+j)}\left(-\frac{t}{N}\right)\right]\underset{N\to\infty}{\longrightarrow}
\frac{d}{dt}\log \det\left[I_{i+j +1}(2\sqrt{t})\right] -\frac{k^2}{t}.
\end{equation}
\begin{remark}\label{FW} Formula (\ref{HIIIxi}), in slightly different but 
equivalent form, was first  obtained by  Forrester and Witte in 
\cite{ForresterW:2006}. If one could prove the asymptotic relation 
(\ref{HkHIII2}) directly, then our second main result, i.e.
Theorem \ref{Theorem2}, may be obtained via a simple reference to 
\cite{ForresterW:2006}.  However, a direct asymptotic
analysis of the Laguerre-Hankel determinant is not a simple matter; 
one can easily see, for instance,  that all the matrix entries have  the same 
leading behavior as $N\to \infty$.  As it stands at the moment,
(\ref{HkHIII2}) is a non-trivial by-product of our 
Riemann-Hilbert analysis.  We note that \cite{BBBCPRS} employs a different 
analysis which does lead directly to (\ref{HkHIII2}) and so does enable the 
results of \cite{ForresterW:2006} to be applied straightforwardly.  But 
this method is not so well-suited to giving exact formulae for finite $N$, 
so should be considered complementary to ours. 
We also note that the confluence of confluent hypergeometric function 
solutions of Painlev\'e V to Bessel function solutions of Painlev\'e III has 
been studied by Masuda \cite{Masuda:2004}.  It would be interesting to see
if the degeneration method of Masuda could provide an alternative proof of 
\eqref{HkHIII2}, which  does not seem to be an immediate corollary of the 
constructions of \cite{Masuda:2004}.
\end{remark}

\section{Conformal block expansion of the \texorpdfstring{$\tau$}{TEXT}-function \label{conformal}}
The    function  $\tau_L$ introduced in \cite{Lisovyy} is defined as 
\beq
\label{tau_L}
x\frac{d}{dx}\log\tau_L := \sigma_L+\frac{\theta_*\lb x+\theta_*\rb}{2},
\eeq
where $\sigma_L\lb x\rb$ satisfies the $\sigma$-form of the Painlev\'e V equation
  \beq\label{Ls}
  \lb x\frac{d^2\sigma_L}{dx^2}\rb^2=\lb \sigma_L-x\frac{d\sigma_L}{dx}+2\lb\frac{d\sigma_L}{dx}\rb^2\rb^2-\tfrac14\lb\lb2 \frac{d\sigma_L}{dx}-\theta_*\rb^2-4\widetilde{\theta}_0^2\rb \lb \lb2 \frac{d\sigma_L}{dx}+\theta_*\rb^2-4\theta_t^2\rb,
  \eeq
  where $\theta_*$, $\theta_t$, and $\widetilde{\theta}_0$ are complex parameters.
  To identify the above equation with the $\sigma$-form in \eqref{sigma_jeq}, we need to make the shift 
  \beq
  \label{dts}
\sigma=\widetilde{\sigma}+x\dfrac{(2\theta_0+\theta_{\infty})}{4}+\dfrac{(2\theta_0+\theta_{\infty})^2}{8}
\eeq
so that 
\begin{equation} 
\label{ts}
  \left(x\frac{d^2\widetilde{\sigma}}{dx^2}\right)^2 = \left[ \ts - x \frac{d\ts}{dx} + 2\left(\frac{d\ts}{dx}\right)^2\right]^2 
  - \frac{1}{4} \left( \left(2\frac{d\ts}{dx}+\frac{\theta_{\infty}}{2}\right)^2-\theta_0^2\right) \left(\left(2\frac{d\ts}{dx}-\frac{\theta_{\infty}}{2}\right)^2-\theta_1^2\right).
\end{equation}
Comparing \eqref{Ls} and \eqref{ts} we have that  $\ts=\sigma_L$ if
\beq
2\theta_*=\theta_{\infty},\quad 4\theta_t^2=\theta_0^2,\quad 4\widetilde{\theta}_0^2=\theta_1^2
\eeq
or
\beq
\label{choice}
2\theta_*=-\theta_{\infty},\quad 4\theta_t^2=\theta_1^2,\quad 4\widetilde{\theta}_0^2=\theta_0^2.
\eeq
Next, we consider the  relations between the Jimbo-Miwa  $\tau$-function  \eqref{tau} and $\tau_L$. By \eqref{Ls} and \eqref{dts} we have
\beq
x\dfrac{d}{dx}\log \tau=x\dfrac{d}{dx}\log\tau_L-\dfrac{x}{4}(2\theta_*+\theta_{\infty})+\frac{\theta^2_0+\theta_1^2}{4}-\frac{\theta_{\infty}^2}{8}-\dfrac{\theta_*^2}{2}.
\eeq
We choose $2\theta_*=-\theta_{\infty}$ so that the relation between the two $\tau$-functions becomes
\beq
\tau(x)=c\;\tau_L(x)x^{\frac{\theta^2_0+\theta_1^2-\theta^2_{\infty}}{4}}
\eeq
for some constant $c$.
Therefore, the  correspondence of the  set of parameters $(\theta_0,\theta_1,\theta_{\infty})$ and  $(\theta_*, \theta_t, \widetilde{\theta}_0)$  is as in \eqref{choice}. There is still an ambiguity in identifying the sign of the parameters $\theta_t$ and $\widetilde{\theta}_0$,
but this is not important for our purpose because the conformal block expansion is symmetric with respect to $\theta_t\to -\theta_t$ and $\widetilde{\theta}_0\to -\widetilde{\theta}_0$.

From \eqref{TP17}, the parameters of the Painlev\'e equation we are considering are
\beq
\theta_*=-\frac{1}{2}\theta_{\infty}=k,\quad  \theta_t=-\frac{1}{2}\theta_1=\frac{k+N}{2}, \quad \widetilde{\theta}_0=\frac{1}{2}\theta_0=\frac{k+N}{2}.
\eeq
Comparing \eqref{zz3}, \eqref{zz6}, and \eqref{f}
we have the relation between the functions $f_k$ and $\tau_L$
\beq
\label{ftau}
x\dfrac{d}{dx}\log f_k=x\frac{d}{d x} \log\tau_L+\frac{(k+N)^2}{2}-k^2-\frac{k}{2}x.
\eeq

 Next, we present the conformal block expansion of the function  $\tau_L$  near $x=0$ as developed in \cite{Lisovyy}.
For this purpose we introduce for any positive integer $N$
the  partition 
\beq
\lambda:=\left\{\lambda_1\geq\lambda_2\geq\ldots\geq\lambda_N>0\right\}.
\eeq
Partitions can be identified in the obvious way with Young diagrams.  The set of all Young diagrams will
 be denoted by $\mathbb{Y}$. For $\lambda\in\mathbb{Y}$, $\lambda'$ denotes the transposed diagram, $\lambda_i$
 and $\lambda'_j$ the number of boxes in the $i$th row and $j$th column of $\lambda$, and $|\lambda|$
 the total number of boxes. Given a box $(i,j)\in\lambda$, its hook length is defined as $h_{\lambda}(i,j):=
 \lambda_i+\lambda'_j-i-j+1$, and for the empty partition, $h_{\emptyset}(i,j)=1$.
For complex numbers $\theta_*, \;\widetilde{\theta}_0$, $\theta_t$, and $\sigma$   and partitions $\lambda$ and $\mu$  let us introduce the quantity
\beq
\label{CBPV2}
\begin{split}
 \mathcal{B}_{\lambda,\mu}\left(\widetilde{\theta}_0,\theta_t,\theta_*,\sigma\right) := &
\prod_{(i,j)\in\lambda}
 \frac{\left(\theta_*+\sigma+i-j\right)\left(\left(\theta_t+\sigma+i-j\right)^2-\widetilde{\theta}_0^2\right)}{
 h_{\lambda}^2(i,j)\left(\lambda'_j+\mu_i-i-j+1+2\sigma\right)^2} \\
 & \times \prod_{(i,j)\in\mu}
 \frac{\left(\theta_*-\sigma+i-j\right)\left(\left(\theta_t-\sigma+i-j\right)^2-\widetilde{\theta}_0^2\right)}{
 h_{\mu}^2(i,j)\left(\lambda_i+\mu'_j-i-j+1-2\sigma\right)^2}\,.
\end{split}
\eeq
 \begin{theorem}[\cite{Lisovyy}]
 The $\tau$-function  of the Painlev\'e V equation  has the following expansion near $x=0$:
 \beq
 \label{fourier0}
 \tau_L\lb x\rb =\mathcal N_0\sum_{n\in\mathbb Z} e^{2\pi in\eta}
 \mathcal C_0\lb \theta_t;  \widetilde{\theta}_0;\theta_*;\sigma+n\rb \mathcal B\lb \theta_t; \widetilde{\theta}_0;\theta_*;\sigma+n; x\rb,
 \eeq
 where  $\sigma,\eta$ correspond to the initial conditions, $\mathcal N_0$ is a constant, $\mathcal B\lb \theta_t; \widetilde{\theta}_0;\theta_*;\sigma;x\rb$ is given by the combinatorial series 
\beq
\mathcal B\lb \theta_t; \widetilde{\theta}_0;\theta_*;\sigma;x\rb :=x^{\sigma^2-\widetilde{\theta}_0^2-\theta_t^2}e^{-\theta_tx}
\sum_{\lambda,\mu\in\mathbb{Y}}
 \mathcal{B}_{\lambda,\mu}\left(\widetilde{\theta}_0,\theta_t,\theta_*,\theta_*;\sigma\right)   x^{|\lambda|+|\mu|}
\eeq
 with $\mathcal{B}_{\lambda,\mu}\left(\widetilde{\theta}_0,\theta_t,\theta_*,\sigma\right)$ as in \eqref{CBPV2},
 and the structure constants $\mathcal C_0\lb \theta_t; \widetilde{\theta}_0;\theta_*;\sigma\rb$ are expressed in terms of the Barnes $G$-function as
 \beq\label{strfourier0}
 \mathcal C_0\lb \theta_t; \widetilde{\theta}_0;\theta_*;\sigma\rb:=\prod_{\epsilon=\pm 1}\frac{G\lb 1+\theta_*+\epsilon\sigma\rb G\lb 1+\widetilde{\theta}_0+\theta_t+\epsilon\sigma\rb G\lb 1-\widetilde{\theta}_0+\theta_t+\epsilon\sigma\rb}{G\lb 1+2\epsilon\sigma\rb},
 \eeq
 where $\theta_*=-\frac{1}{2}\theta_{\infty}$, $\theta_t=-\frac{1}{2}\theta_1$, and $\widetilde{\theta}_0=\frac{1}{2}\theta_0$.
  \end{theorem}
 Comparing the expansion near $x=0$  of $\tau_L$ in \eqref{fourier0} with \eqref{ftau}  and \eqref{f},  and using the fact that 
\beq
\sum_{\lambda,\mu\in\mathbb{Y}}
 \mathcal{B}_{\lambda,\mu}\left(\widetilde{\theta}_0,\theta_t,\theta_*,\sigma\right)   x^{|\lambda|+|\mu|}=1+\left(\theta_{t}+\frac{1}{2}\theta_*+\frac{\theta_*}{2\sigma^2}(\theta_t^2-\widetilde{\theta}_0^2)\right)x+{\mathcal O}(x^2),
\eeq
 we obtain  that  $\sigma=k$ and $\eta=0$ and, furthermore, the sum in \eqref{fourier0}  is only over non-negative integers $n$.
 Unfortunately, for the  values $\sigma=k$, $\theta_*=k$, $\theta_t=\widetilde{\theta}_0=\frac{k+N}{2}$,
  the structure constants $\mathcal C_0\lb \theta_t; \widetilde{\theta}_0;\theta_*;\sigma\rb$  are undefined.
  For this reason we need first to take a limit, using the following  relation  for the Barnes $G$ function  that holds for non-negative integers $n$:
\beq
  G(1+\delta-n)=\delta^n(-1)^{\frac{n(n-1)}{2}}G(1+n)+O(\delta^{n+\epsilon}),\quad n\in\mathbb{N}_0,\epsilon,\delta>0.
\eeq
  Then we have that, for $\sigma=k+n$ a non-negative integer,
  \begin{equation}
  \label{Ctilde}
  \begin{split}
\widetilde{\mathcal C}_0\lb\theta_t;k;\sigma\rb& :=  \lim_{\delta\to 0} 2^{-2k}(-1)^{\frac{k(k+1)}{2} }\delta^k\mathcal C_0\lb\theta_t; \theta_t;k;\sigma-\delta\rb\\
&=
\frac{G(1+k+\sigma)G(1+2\theta_t+\sigma)G(1+2\theta_t-\sigma)G(1+\sigma)^2}{(-1)^{\sigma(\sigma+k)}2^{2\sigma-2k}G(1+2\sigma)^2},
  \end{split}
  \end{equation}
  where we observe that we have the freedom to multiply the structure constants by $\sigma$-independent quantities.
  Furthemore, the $\tau_L$ function  is defined up to the constant ${\mathcal N}_0$ that we obtained from \eqref{z9}, namely, we must have
\beq
  F_N(0,k)=\dfrac{G(N+2k+1)G(N+1)G(k+1)^2}{G(N+k+1)^2G(2k+1)}=f_k(0)={\mathcal N}_0\widetilde{\mathcal C}_0\lb  \frac{N+k}{2};k;k\rb,
\eeq
  which implies
\beq
  {\mathcal N}_0=\dfrac{1}{G(N+k+1)^2}.
\eeq
We also observe that $\widetilde{\mathcal C}_0\lb \frac{N+k}{2};k;N+k+1\rb=0$, namely, we have only $N+1$ conformal blocks.
  We arrive at the following conjectural expression.
    \begin{conjecture}
  The function $f_k(x)$ defined in \eqref{a5} has the following conformal block expansion near $x=0$:
\begin{equation}
\label{f_CB}
\begin{split}
f_k(x) = & \dfrac{e^{-\frac{N}{2}x-kx}}{G(N+k+1)^2} \sum_{n=0}^{N}\widetilde{\mathcal C}_0\lb \frac{N+k}{2};k;k+n\rb x^{2nk+n^2} \\ 
  & \times \sum_{\lambda,\mu\in\mathbb{Y}}
 \mathcal{B}_{\lambda,\mu}\left(\frac{k+N}{2},\frac{k+N}{2},k,k+n\right)   x^{|\lambda|+|\mu|},
\end{split}
\end{equation}
where the coefficients $\widetilde{\mathcal C}_0\lb \frac{N+k}{2};k;k+n\rb$ are defined in \eqref{Ctilde} and the conformal blocks are defined in \eqref{CBPV2}.
\end{conjecture}
The coefficients 
$\beta_{2j}$ in \eqref{zz22}  for $j=1,\dots, k$  of the power series expansion of $f_k(x)$ near zero are  obtained from the   first  conformal block  $\mathcal{B}_{\lambda,\mu}\left(\frac{k+N}{2},\frac{k+N}{2},k,k\right)$.  This conformal block contains the term $(i-j)$ in the sum over the partition $\mu$ and therefore it is nonzero only for the empty partition.
 Furthermore, the factor $(k+i-j)$  in the product over boxes of $\lambda$ reduces the summation  in \eqref{f_CB} to Young diagrams with $\lambda_1\leq k$. 
Therefore, the sum over the first conformal block reduces to 
\begin{equation}
\label{f_CB1}
\begin{split}
f_k(x)&=\dfrac{\widetilde{\mathcal C}_0\lb \frac{N+k}{2};k;k\rb}{G(N+k+1)^2}e^{-\frac{N}{2}x-kx}\sum_{\lambda,\mu\in\mathbb{Y}}
 \mathcal{B}_{\lambda,\mu}\left(\frac{k+N}{2},\frac{k+N}{2},k,k\right)x^{|\lambda|+|\mu|} +\dots\\
&=f_k(0)e^{-\frac{N}{2}x-kx}\sum_{\overset{\lambda\in\mathbb{Y},}{\lambda_1\leq k}}\prod_{(i,j)\in\lambda}
 \frac{\left(2k+i-j\right)\left(N+2k+i-j\right)(k+i-j)}{ h_{\lambda}^2(i,j)\left(\lambda'_j-i-j+1+2k\right)^2}x^{|\lambda|}+\dots \\
  & =f_k(0)( 1 + b_2x^2 +b_4x^4 + b_6x^6 + b_8x^8 +b_{10}x^{10} + \dots),
\end{split}
\end{equation}
where  the coefficients $b_{2j}$ coincides with the coefficients $\beta_{2j}$ defined in  \eqref{zz22} and \eqref{beta}, namely
\beq
\begin{split}
b_2 := & \frac{N  (2 k+N)}{8-32 k^2}=\beta_2, \quad b_4 := \frac{N (2 k+N) \left(2 k N+N^2+2\right)}{128 \left(16 k^4-40 k^2+9\right)}=\beta_4, \\
b_6 := & -\frac{N (2 k+N)}{3072 \left(\left(1-4 k^2\right)^2 \left(16 k^4-136k^2+225\right)\right)} \\ 
  & \times (16 k^4 N^2+16 k^3 N \left(N^2+3\right)+4 k^2 \left(N^4-3 N^2+16\right) \\ 
  & \hspace{.2in}+4 k N \left(5-9 N^2\right)-9 N^4+10N^2-16)=\beta_6,\\
b_8 := & \frac{N (2 k+N) \left(2 k N+N^2+6\right) } {98304 \left(1-4 k^2\right)^2 (64 k^6-1328 k^4+7564k^2-11025)} \\
  & \times (16 k^4 N^2+16 k^3 N \left(N^2+3\right)+4k^2 \left(N^4-27 N^2+40\right) \\ 
  & \hspace{.2in} +4 k N (29-33 N^2)-33 N^4+58N^2-40)=\beta_8,\\
b_{10} := & -\frac{N (2 k+N)}{3932160 (4k^2-81) (4 k^2-49) (4 k^2-25) (4 k^2-9)^2 (4 k^2-1)^2 } \\ 
  & \times (256 k^8 N^4+512 k^7 N^5+2560 k^7 N^3+384 k^6 N^6-1920 k^6 N^4+14080 k^6 N^2 \\ 
  & \hspace{.2in} +128 k^5 N^7-9600 k^5 N^5-23040 k^5 N^3+48640 k^5 N+16 k^4 N^8-8320 k^4 N^6\\
   & \hspace{.2in} -28208 k^4 N^4-67200 k^4 N^2+86016 k^4-2880 k^3 N^7+20064 k^3 N^5 \\ 
   & \hspace{.2in} -20960 k^3 N^3-54016 k^3 N-360 k^2 N^8 +31288 k^2 N^6+82960 k^2 N^4 \\ 
   & \hspace{.2in} +70768 k^2 N^2-215040 k^2+11976 k N^7+52920 k N^5 +97776 k N^3 \\ 
   & \hspace{.2in} -149280 k N+1497 N^8+8820 N^6+24444 N^4-74640 N^2+48384)=\beta_{10}.
\end{split}
\eeq
We remark that the combinatorial expression for the coefficient $b_{2k}=\beta_{2k}$ provided by the conformal block expansion \eqref{f_CB1} should be consistent with the combinatorial expression
 obtained in \cite{Dehaye:2008}.
 
Finally, we want to consider the limit $N\to\infty$ as done in 
\cite{Lisovyy1} that reduces the $\tau_L$ function of  the Painlev\'e V 
equation  to the $\tau$ function of the  Painlev\'e III equation.  In our 
case  we have that the quantity
\beq
\dfrac{ f_k(\frac{t}{N})}{N^{k^2}} 
\eeq
has a well-defined term-by-term limit as $N\to\infty$  that can be  easily obtained using the properties  of the Barnes $G$-function (see Appendix~\ref{Barnes}) and 
the fact that  $\mathcal{B}_{\lambda,\mu}\left(\frac{k+N}{2},\frac{k+N}{2},k,k\right)\left(\frac{t}{N}\right)^{|\lambda|+|\mu|} $ has a well-defined limit that can be calculated term-wise.
 Therefore, we can define the function
\beq
\tau_{III}(t):=\lim_{N\to\infty}\dfrac{ f_k(\frac{t}{N})}{N^{k^2}}.
\eeq
The function $\tau_{III}(t)$ has a conformal block expansion 

 \begin{equation}\label{tf}
\tau_{III}(t)=e^{-\frac{t}{2}}\sum_{n=0}^{\infty}
 {\mathcal C}_{III}(k;k+n) t^{2nk+n^2}\sum_{\lambda,\mu\in\mathbb{Y}}
 \mathcal{B}^{III}_{\lambda,\mu}\left(k,k+n\right)   t^{|\lambda|+|\mu|},
\end{equation}
where 
\beq
{\mathcal C}_{III}(k;\sigma) =\frac{G(1+k+\sigma)G(1+\sigma)^2}{(-1)^{\sigma(\sigma+k)}2^{2\sigma-2k}G(1+2\sigma)^2},
\eeq
and 
\begin{equation}
\label{CBPIII}
\begin{split}
{ \mathcal B}^{III}_{\lambda,\mu}\left(\theta_*,\sigma\right) = &
\prod_{(i,j)\in\lambda}
 \frac{ \left(\theta_*+\sigma+i-j\right)(\sigma+i-j)}{
 h_{\lambda}^2(i,j)\left(\lambda'_j+\mu_i-i-j+1+2\sigma\right)^2} \\
 & \times \prod_{(i,j)\in\mu}
 \frac{\left(\theta_*-\sigma+i-j\right)\left(-\sigma+i-j\right)}{
 h_{\mu}^2(i,j)\left(\lambda_i+\mu'_j-i-j+1-2\sigma\right)^2}\,.
\end{split}
\end{equation}

The function $\xi(t):=t\dfrac{d}{dt}\log  \tau_{III}(t)$ satisfies the $\sigma$-form of the Painlev\'e III equation \eqref{sl40} 
 with boundary conditions $\xi(0)=0$ and $\xi'(0)=0$.
We write the first few terms of the expansion of $\tau_{III}(t)$ near $t=0$:
\begin{align}
\tau_{III}(t)&=\dfrac{G(1+k)^2}{G(1+2k)}
\left[ -\frac{t^2}{4(4k^2-1)2!}+ \frac{t^4}{16(4k^2-1)(4k^2-9)4!}+\right.\\
&\left.- \frac{15t^6}{64(4k^2-1)^2(4k^2-25)6!}+ t^8\frac{105(4k^2-33)}{256(4k^2-1)^2(4k^2-9)(4k^2-25)(4k^2-49)8!}\right.\\
&-\left. t^{10}\frac{925(16k^4-360k^2+1497)}{1024(4k^2-1)^2(4k^2-9)^2(4k^2-25)(4k^2-49)(4k^2-81)10!}+\mathcal{O}(t^{11})\right]\,.
\end{align}
One can observe that the function $\tau_{III}(t)$ exactly reproduces the  coefficients $F(h,k)$.  In particular, we observe that it is sufficient to consider only the first conformal
block to obtain the coefficients $F(h,k)$ for $k>h-\frac{1}{2}$.  The quantity $\mathcal{B}^{III}_{\lambda,\mu}(k,k)$  vanishes for any non-empty partition $\mu$.  Therefore, it is sufficient to sum only over  Young diagrams  $\lambda$.
 Furthermore,
the factor $(k+i-j)$  in the product over boxes of $\lambda$ reduces the summation  in \eqref{CBPIII} to Young diagrams with $\lambda_1\leq k$. 
\begin{conjecture}
We have the following conjectural  relation  for the function $F(h,k)$ defined in \eqref{Fdef}:
\begin{equation}
\begin{split}
F(h,k)&=(-1)^{h}\dfrac{G(k+1)^2}{G(2k+1)}(2h)!\sum\limits_{\overset{\lambda\in\mathbb{Y}}{|\lambda|=2h,\lambda_1\leq k}}\prod_{(i,j)\in\lambda}
 \frac{ \left(2k+i-j\right)(k+i-j)}{
 h_{\lambda}^2(i,j)\left(\lambda'_j-i-j+1+2k\right)^2},
\end{split}
\end{equation}
with $k>h-\frac{1}{2}$.
\end{conjecture}
In  \cite{Dehaye:2008} a combinatorial formula for $F(h,k)$ has been obtained and the first few terms of this formula are compatible with ours.

  \vskip 0.5cm

\appendix

\section{Differential-difference identities and the proof of Lemma \ref{Hn-m-lemma}}
\label{lemma-appendix}

\numberwithin{equation}{section}

The jump matrix of the $\Psi$-Riemann-Hilbert problem  (\ref{PV17})--(\ref{PV19}) is not only independent of $z$ and $x$, but 
also independent of $n$. Therefore, if we indicate the dependence of $\Psi(z)$ on $n$ as
\beq
\Psi(z) \equiv \Psi_n(z),
\eeq
we can state that the discrete logarithmic derivative $\Psi_{n+1}\Psi^{-1}_{n}(z)$ is also analytic 
in $\mathbb{C} \setminus ( \{0\} \cup\{1\} )$. In fact, since the singular right factors in the right-hand sides
of formulae (\ref{PV22}) and  (\ref{PV19}) are also $n$-independent, we conclude that the only
singularity of $\Psi_{n+1}\Psi^{-1}_{n}(z)$ is the simple pole at $z=0$. Therefore, we conclude that, 
in addition to the two differential equations (\ref{PV23}) and (\ref{PV24}), the function $\Psi(z,x)\equiv 
\Psi_{n}(z,x)$ satisfies the difference equation
\beq\label{differequ}
\Psi_{n+1}(z) = \left( \frac{1}{z} U_{-1} + U^{(n)}_0\right)\Psi_n(z) \equiv U_n(z)\Psi_n(z).
\eeq
Let us determine the structure of the matrix coefficients $U_{-1}$ and $U^{(n)}_0$ in (\ref{differequ}). 
We consider a more detailed form of the expansion (\ref{PV20}),
\beq\label{PV41}
\Psi_n (z)=     \begin{pmatrix}
(-1)^n & 0 \\
0 &  (-1)^{k+N-n}
\end{pmatrix}
      \left( I+zM_1     
 +\cdots  \right)  \begin{pmatrix}
z^{-n} & 0 \\
0 &  z^{n-3k-N}
\end{pmatrix}, \quad z \rightarrow 0,
\eeq
where
\beq\label{PV42}
M_1\equiv \begin{pmatrix}
a_n & b_n\\
c_n& d_n
\end{pmatrix} =
\begin{pmatrix}
-n+\frac{x}{2} & 0 \\
0 & n-k-N-\frac{x}{2}
\end{pmatrix} - m^{(n)}_1.
\eeq
Here $m^{(n)}_1$ is the matrix coefficient from the expansion  (\ref{PV110}) 
with the explicit indication of the dependence on the integer  $n$. Remembering 
the relation of the coefficient $m_1$ with the norm $h_n$ of the orthogonal polynomials $P_n(z)$ (see  (\ref{PV28})), we note that
\beq\label{PV47}
\begin{split}
a_n=-n+\frac{x}{2}-(m_1^{(n)})_{11}, \quad b_n=-(m_1^{(n)})_{12}=\frac{1}{2\pi i} h_n, \hspace{.25in} \\
c_n=-(m_1^{(n)})_{21}=\frac{2\pi i}{h_{n-1}}, \quad d_n=n-k-N-\frac{x}{2}-(m_1^{(n)})_{22}.
\end{split}
\eeq
Observe that, in particular,
\beq\label{PV471}
b_n c_{n+1} =1.
\eeq
Plugging (\ref{PV41}) into the right-hand side of the equation $U_n=\Psi_{n+1} \,\Psi_n^{-1}$, we have that
\beq\label{PV45}
\begin{split}
U_n = & \begin{pmatrix}
(-1)^{n+1} & 0 \\
0 &  (-1)^{k+N-n-1}
\end{pmatrix}
      \left( I+zM_1^{(n+1)}     
 +\cdots  \right)  \begin{pmatrix}
z^{-n-1} & 0 \\
0 &  z^{n+1-3k-N}
\end{pmatrix} \\
 & \times  \begin{pmatrix}
z^{n} & 0 \\
0 &  z^{-n-3k-N}
\end{pmatrix} 
 \left( I-zM_1^{(n)}     
 +\cdots  \right) 
 \begin{pmatrix}
(-1)^n & 0 \\
0 &  (-1)^{k+N-n}
\end{pmatrix} \\
= & -\frac{1}{z}\begin{pmatrix}
1& 0 \\
0 &  0
\end{pmatrix}  +
\begin{pmatrix}
a_n-a_{n+1} & (-1)^{N+k}b_n \\
-(-1)^{N+k}c_{n+1} &  0
\end{pmatrix}  + \mathcal{O}(z), \quad z \rightarrow 0.
\end{split}
\eeq
Hence,
\beq\label{U1U0}
U_{-1} = -\begin{pmatrix}
1& 0 \\
0 &  0
\end{pmatrix}
\eeq
and
\beq\label{U1U02}
U^{(n)}_0 = \begin{pmatrix}
a_n-a_{n+1} & (-1)^{N+k}b_n \\
-(-1)^{N+k}c_{n+1} &  0
\end{pmatrix}.
\eeq
Let us rewrite the Lax pair  (\ref{PV23})--(\ref{PV24})  as 
\beq\label{PV23n}
\frac{\partial\Psi_n}{\partial z}=\left( xA^{(n)}_\infty +\frac{A^{(n)}_0}{z}+\frac{A^{(n)}_1}{z-1} \right)\Psi_n \equiv A_n(z)\Psi_n(z),
\eeq
\beq\label{PV24n}
\frac{\partial\Psi_n}{\partial x}=zA^{(n)}_\infty \Psi_n \equiv B_n(z)\Psi_n(z),
\eeq
indicating explicitly the dependence of all objects on $n$. 
The compatibility of the equations (\ref{PV23n}) and (\ref{PV24n}) with the difference equation  (\ref{differequ}) 
 yields the differential-difference zero-curvature equations
\beq\label{PV51}
\frac{\partial U_n}{\partial z}=A_{n+1}(z)U_n -U_nA_n(z)
\eeq
and 
\beq\label{PV512}
\frac{\partial U_n}{\partial x}=B_{n+1}(z)U_n -U_nB_n(z).
\eeq
From (\ref{PV51}) and (\ref{PV512}) it follows that 
\beq\label{PV52}
A_\infty^{(n+1)}U_0^{(n)} -U_0^{(n)}A_\infty^{(n)}=0
\eeq
and
\beq\label{PV53}
\frac{d}{dx} U_0^{(n)}=   \begin{pmatrix}
1& 0 \\
0 &  0
\end{pmatrix} A_\infty^{(n)}-A_\infty^{(n+1)}
\begin{pmatrix}
1& 0 \\
0 &  0
\end{pmatrix}.
\eeq
Defining the matrix $A_\infty$ as
\beq\label{PV54}
A_\infty :=\begin{pmatrix}
\alpha_n& \beta_n\\
\gamma_n &  -\alpha_n
\end{pmatrix} 
\eeq
and using (\ref{U1U02}),
we have from (\ref{PV52}) that 
$$(a_n-a_{n+1})(\alpha_{n+1}-\alpha_n)=(-1)^{N+k}(\beta_{n+1} +\gamma_nb_n), \quad \gamma_{n+1}b_n+c_{n+1}\beta_n=0,$$
\beq\label{PV56}
\alpha_{n+1}b_n(-1)^{N+k}-\beta_n(a_n-a_{n+1})+(-1)^{N+k}b_n\alpha_n=0,
\eeq
$$\gamma_{n+1}(a_n-a_{n+1})+(-1)^{N+k}\alpha_{n+1}c_{n+1}+(-1)^{N+k}\alpha_nc_{n+1}=0.$$
Also, from (\ref{PV53}),
\beq\label{pv57}
(-1)^{N+k}\frac{d\,b_n}{dx}=\beta_n
\eeq
or
\beq\label{PV58}
\frac{1}{b_n}\frac{d\,b_n}{dx}=(-1)^{N+k}\beta_n\frac{1}{b_n}.
\eeq
Using (\ref{PV471}), one also obtains that
\beq\label{PV59}
\frac{d}{dx} \log b_n=(-1)^{N+k}c_{n+1}\beta_n.
\eeq
Recall now the relation between the functions $\Psi_n(z)$ and $\Phi_n(z)$,
\beq\label{T5}
\Psi_n(z)=Y^{(n)}(1)\Phi_n(z)z^{-\frac{3k+N}{2}}  (z-1)^{\frac{k+N}{2}}.
\eeq
Taking into account (\ref{TP14}), we conclude  from (\ref{T5}) that,
as $z\rightarrow \infty$,
\beq\label{T6}
\begin{split}
\Psi_n(z) & = Y^{(n)}(1)\left( I+\frac{\phi^{(n)}_1}{z}+\cdots \right) e^{\frac{xz}{2}\sigma _3 } \begin{pmatrix}
z^k & 0 \\
0 &  z^{-k}
\end{pmatrix} z^{-k}\left(1-\frac{N+k}{2}\frac{1}{z}+\cdots\right) \\
& = Y^{(n)}(1)\left( I+\frac{1}{z}\left(\phi^{(n)}_1-\frac{N+k}{2} I\right)+\cdots \right) e^{\frac{xz}{2}\sigma _3 } \begin{pmatrix}
1& 0 \\
0 &  z^{-2k}
\end{pmatrix} \\
& = Y^{(n)}(1)\left( I+\frac{1}{z}{\widehat{\phi}}^{(n)}_1 +\cdots \right) e^{\frac{xz}{2}\sigma _3 } \begin{pmatrix}
1 & 0 \\
0 &  z^{-2k}
\end{pmatrix},
\end{split}
\eeq
where
\beq\label{mhat}
{\widehat{\phi}}^{(n)}_1:=\phi^{(n)}_1-\frac{N+k}{2} I.
\eeq
Let us plug this expansion  into the difference equation (\ref{differequ}):
\beq\label{T9}
Y^{(n+1)}(1)\left(I+\frac{1}{z}{\widehat{\phi}}_1^{(n+1)} +\cdots \right) =\left[-\frac{1}{z}\begin{pmatrix}
1& 0 \\
0 &  0
\end{pmatrix}  + U_0^{(n)}
\right] 
Y^{(n)}(1)
\left(I+\frac{1}{z}{\widehat{\phi}}_1 ^{(n)}+\cdots \right).
\eeq
We then see that
\beq\label{T10}
Y^{(n+1)}(1)=U_n^{(0)}Y^{(n)}(1)
\eeq
and
\beq\label{T11}
Y^{(n+1)}(1){\widehat{\phi}}_1^{(n+1)}=U_n^{(0)}Y^{(n)}(1){\widehat{\phi}}_1^{(n)}-\begin{pmatrix}
1& 0 \\
0 &  0
\end{pmatrix} Y^{(n)}(1).
\eeq
Excluding $Y^{(n+1)}(1)$, we arrive at the formula
\beq\label{T12}
{\widehat{\phi}}_1 ^{(n+1)} - {\widehat{\phi}}_1 ^{(n)}\equiv 
\phi_1 ^{(n+1)} - \phi_1 ^{(n)}=-\left[Y^{(n)}(1)\right]^{-1}\left[U_n^{(0)}\right]^{-1}
\begin{pmatrix}
1& 0 \\
0 &  0
\end{pmatrix} Y^{(n)}(1).
\eeq
Put
\beq\label{T13}
Y^{(n)}(1):=\begin{pmatrix}
p_n& q _n\\
r_n &  s_n
\end{pmatrix}.
\eeq
Note that since $\det Y^{(n)}(1)=1$ we have
\beq
p_ns_n - q_nr_n = 1.
\eeq
With these notations and recalling (\ref{U1U02}), we have from (\ref{T12}) 
that
\beq\label{T14}
\begin{split}
\phi_1 ^{(n+1)} & - \phi_1 ^{(n)} \\  
& = -\begin{pmatrix}
s_n& -q_n \\
-r_n &  p_n
\end{pmatrix}\begin{pmatrix}
0& -(-1)^{N+k}b_n \\
(-1)^{N+k}c_{n+1}&  a_n-a_{n+1}
\end{pmatrix}\begin{pmatrix}
1& 0 \\
0 &  0
\end{pmatrix} \begin{pmatrix}
p_n& q_n \\
r_n &  s_n
\end{pmatrix} \\
& = - \begin{pmatrix}
-q_n(-1)^{N+k}c_{n+1}& -(-1)^{N+k} s_n b_n+q_n(a_{n+1}-a_n \\
(-1)^{N+k} p_n c_{n+1}&  r_n b_n(-1)^{N+k} +p_n (a_n-a_{n+1})
\end{pmatrix} \begin{pmatrix}
1& 0 \\
0 &  0
\end{pmatrix} \begin{pmatrix}
p_n& q_n \\
r_n &  s_n
\end{pmatrix} \\
& = - \begin{pmatrix}
-(-1)^{N+k} q_n c_{n+1}& 0\\
(-1)^{N+k} p_n c_{n+1}&  0
\end{pmatrix} \begin{pmatrix}
p_n& q_n \\
r_n &  s_n
\end{pmatrix}.
\end{split}
\eeq
In particular,
\beq\label{T16}
(\phi_1^{(n+1)})_{11} - (\phi_1^{(n)})_{11}=(-1)^{N+k} p_n q_n c_{n+1}.
\eeq
On the other hand, from (\ref{PV26}) it follows that
\beq\label{T17}
A^{(n)}_{\infty}=\frac{1}{2}\begin{pmatrix}
p_n& q_n \\
r _n&  s_n
\end{pmatrix} \begin{pmatrix}
1& 0 \\
0 &  -1
\end{pmatrix} \begin{pmatrix}
s_n & -q_n \\
-r_n &  p_n
\end{pmatrix}
=\frac{1}{2}\begin{pmatrix}
p_ns_n+qr& -2p_nq_n \\
-2p_nr_n&  -s_np_n-r_nq_n
\end{pmatrix}.
\eeq
Comparing these equations with (\ref{PV54}), we conclude that
\beq\label{T19}
\beta_n=-p_n q_n.
\eeq
Therefore, (\ref{T16}) can be rewritten as 
\beq\label{T20}
(\phi_1^{(n+1)})_{11} - (\phi_1^{(n)})_{11}=-(-1)^{N+k} \beta_n  c_{n+1}
\eeq
which, together with ({\ref{PV59}}), yields the important formula
\beq\label{T21} 
\frac{d}{dx}\log b_n=-(\phi_1^{(n+1)})_{11} + (\phi_1^{(n)})_{11}
\eeq
or, remembering  (\ref{PV47}),
\beq\label{T211} 
\frac{d}{dx}\log h_n=-(\phi_1^{(n+1)})_{11} + (\phi_1^{(n)})_{11}.
\eeq
With equation (\ref{T211}) we are ready to prove Lemma 2. Indeed, taking into account that
\beq\label{DH1}
\log H_n- \log H_1= \log h_{n-1}+\log h_{n-2}+ \cdots + \log h_1,
\eeq
we have from (\ref{T211}) that
\beq\label{DH2}
\frac{d}{dx}\log H_n -\frac{d}{dx}\log H_1= -(\phi_1^{(n)})_{11} + (\phi_1^{(1)})_{11}.
\eeq
Hence, in order to obtain the statement of Lemma \ref{Hn-m-lemma} one only has to calculate explicitly $(\phi_1^{(1)})_{11}$ and  show that
\begin{equation}\label{lemma2proof}
(\phi_1^{(1)})_{11} =-\frac{d}{dx}\log H_1 + \frac{N+k}{2}.
\end{equation}

First, we notice that if we define $\kappa^{(n)}$ as the matrix coefficient in the expansion
\beq\label{DH6}
Y^{(n)}(s)=Y^{(n)}(1)\Bigl( I+ \kappa^{(n)}(s-1)+ \cdots\Bigr), \quad s \rightarrow 1,
\eeq
then by a straightforward calculation we arrive at the relation
\beq\label{DH7}
\phi^{(n)}_{1}=-\kappa^{(n)} +\frac{N+k}{2}\sigma_3.
\eeq
Indeed, we have that (dropping the indication of the dependence of $n$)
\beq\label{DH5}
\begin{split}
\Phi(z)& = Y^{-1}(1)\Psi (z,x)z^{\frac{3k+N}{2}}(z-1)^{-\frac{N+k}{2}} \\
 & = Y^{-1}(1)  \chi (z) e^{\frac{xz}{2}\sigma_3} \begin{pmatrix}
1& 0 \\
0 & z^{( -3k-N)}(z-1)^{N+k}
\end{pmatrix}  z^{\frac{3k+N}{2}}(z-1)^{-\frac{N+k}{2}} \\
 & = Y^{-1}(1)  X (z) e^{\frac{xz}{2}\sigma_3} \begin{pmatrix}
z^{\frac{3k+N}{2}}(z-1)^{-\frac{N+k}{2}}& 0 \\
0 & z^{\frac{( -3k-N)}{2}}(z-1)^{\frac{N+k}{2}}
\end{pmatrix} 
\end{split}
\eeq
and hence, as $z\rightarrow \infty$,
\beq
\begin{split}
\Phi(z) & = Y^{-1}(1)  X (z)  
\left(I+ \frac{N+k}{2}\frac{1}{z}\sigma_3+ \cdots \right) z^{k\sigma_3}
e^{\frac{xz}{2}\sigma_3} \\
 & = Y^{-1}(1)  Y\left(\frac{z-1}{z}\right) 
\left(I+ \frac{N+k}{2}\frac{1}{z}\sigma_3+ \cdots \right) z^{k\sigma_3}
e^{\frac{xz}{2}\sigma_3} \\
 & = Y^{-1}(1)  Y(1) 
\left(I+ \kappa \left(\frac{z-1}{z}-1\right) +\frac{N+k}{2}\frac{1}{z}\sigma_3+ \cdots \right) z^{k\sigma_3}
e^{\frac{xz}{2}\sigma_3} \\
 & = \left[ I-\frac{1}{z}\left(\kappa -\frac{N+k}{2}\sigma_3+ \cdots\right) \right]z^{k\sigma_3}
e^{\frac{xz}{2}\sigma_3},
\end{split}
\eeq
where the last equation, in view of the definition (\ref{Phihatinfty0}) of $\phi_1^{(n)}$, is equivalent to (\ref{DH7}).

Secondly, from the definition (\ref{Ydef0}) of the function $Y^{(n)}(z)$ we have that
\beq\label{DH8}
Y^{(1)}(s)=\begin{pmatrix}
P_1(s)& \frac{1}{2\pi i}\int_C\frac{P_1(s')w_{0}(s') }{s'-s}ds' \\
-\frac{2\pi i}{h_0}P_0(s) & - \frac{1}{h_0}\int_C\frac{w_{0}(s') }{s'-s}ds'
\end{pmatrix},
\eeq
where
\beq
P_0(s)=1,\quad  P_1(s)=s+c,\quad 
c=-\frac{1}{\int _{C}w_{0}(s) ds} \int_{C}  s w_{0}(s) ds = -\frac{1}{h_0}\int_{C}  s w_{0}(s) ds.
\eeq
Hence,
\beq\label{DH9}
Y^{(1)}(1)=\begin{pmatrix}
1+c & \frac{1}{2\pi i}\int_C\frac{(s+c) w_{0}(s) }{s-1}ds \\
-\frac{2\pi i}{h_0} & - \frac{1}{h_0}\int_C\frac{w_{0}(s) }{s-1}ds 
\end{pmatrix},
\quad
Y^{(1)\prime}(1)=\begin{pmatrix}
1& \frac{1}{2\pi i}\int_C\frac{(s+c) w_{0}(s) }{(s-1)^2}ds \\
0 & - \frac{1}{h_0}\int_C\frac{w_{0}(s) }{(s-1)^2}ds 
\end{pmatrix}.
\eeq
For the matrix  coefficient $\kappa$ we have (recall that $\det Y^{(n)} \equiv 1$)
\beq\label{DH11}
\kappa^{(1)}=Y^{(1)-1}(1) Y^{(1)\prime}(1)=\begin{pmatrix}
- \frac{1}{h_0}\int_C\frac{w_{0}(s) }{s-1}ds &-\frac{1}{2\pi i}\int_C\frac{(s+c) w_{0}(s) }{(s-1)^2}ds \\
2\pi i& 1+c
\end{pmatrix}
\begin{pmatrix}
1 & \frac{1}{2\pi i}\int_C\frac{(s+c) w_{0}(s) }{(s-1)^2}ds \\
0& - \int_C\frac{w_{0}(s) }{(s-1)^2}ds.
\end{pmatrix}.
\eeq
Therefore,
\beq\label{DH12}
(\kappa^{(1)})_{11}=- \frac{1}{h_0}\int_C\frac{w_{0}(s) }{s-1}ds 
\eeq
and (cf. (\ref{DH7}))
\beq\label{DH13}
({\phi}_1^{(1)})_{11}=\frac{N+k}{2}+\frac{1}{h_0}\int_C\frac{w_{0}(s) }{s-1}ds 
= \frac{N+k}{2} - \frac{1}{h_0} \int_C\frac{e^{\frac{x}{1-s} }}{(1-s)^{2k+1}s^{N+k}}ds,
\eeq
where
\beq\label{DH14}
h_0=\int_{C}w_0(s)ds = \int_C\frac{e^{\frac{x}{1-s} }}{(1-s)^{2k}s^{N+k}}ds.
\eeq
Now note that 
\beq
H_1 = h_0
\eeq
and 
\beq\label{DH15}
\frac{d}{d x} \log H_1=\frac{1}{H_1} \int_C\frac{e^{\frac{x}{1-s} }}{(1-s)^{2k+1}s^{N+k}}ds 
= \frac{1}{h_0} \int_C\frac{e^{\frac{x}{1-s} }}{(1-s)^{2k+1}s^{N+k}}ds. 
\eeq
Equation (\ref{lemma2proof}) follows from (\ref{DH13}) and (\ref{DH15}). This completes the proof of Lemma \ref{Hn-m-lemma}.

\section{A second  derivation of the second term in the expansion of \texorpdfstring{$\sigma_k(x)$}{TEXT}  }
In this appendix we derive the second term of the expansion \eqref{ex_sigma} directly from the Laguerre determinant.
We write the Laguerre determinant as 
\eq
\label{Laguerre-det}
\mathcal{L}(x):={\mbox{det}} \left[L_{N+k-1-(i+j)}^{(2k-1)} (-x)\right] _{i,j=0,\cdots ,k-1}.
\endeq
Combining this with \eqref{Sum1} and \eqref{sigmadef} gives
\eq
\sigma_k(x) = -Nk + \frac{\mathcal{L}'(x)}{\mathcal{L}(x)}x.
\endeq
Therefore the desired term is 
\eq
\sigma_k'(0) = \frac{\mathcal{L}'(0)}{\mathcal{L}(0)}.
\endeq
Fix the determinant size $k$.  Note from \eqref{laugerre1} that each entry of 
the determinant in \eqref{Laguerre-det} has the structure 
\eq
L_{N+k-1-(i+j)}^{(2k-1)}(-x) = \ell_{i,j} + \frac{N+k-1-(i+j)}{2k}\ell_{i,j}x + \mathcal{O}(x^2), \quad x\to 0.
\endeq
Here $\ell_{i,j}$ (which depends on $N$ and $k$) can be read off from 
\eqref{laugerre1}, although we will not need its particular form.  Now 
$\mathcal{L}(0)$ is simply the matrix with $ij$ entry $\ell_{i,j}$, and 
\eq
\label{Lprime-at-0}
\begin{split}
\mathcal{L}'(0) = & \frac{1}{2k}\begin{vmatrix} (N+k-1)\ell_{0,0} & \ell_{0,1} & \cdots & \ell_{0,k-1} \\ (N+k-2)\ell_{1,0} & \ell_{1,1} & \cdots & \ell_{1,k-1} \\ \vdots & \vdots & & \vdots \\(N+0)\ell_{k-1,0} & \ell_{k-1,1} & \cdots & \ell_{k-1,k-1} \end{vmatrix} + \frac{1}{2k}\begin{vmatrix} \ell_{0,0} & (N+k-2)\ell_{0,1} & \cdots & \ell_{0,k-2} \\ \ell_{1,0} & (N+k-3)\ell_{1,1} & \cdots & \ell_{1,k-1} \\ \vdots & \vdots & & \vdots \\ \ell_{k-1,0} & (N-1)\ell_{k-1,1} & \cdots & \ell_{k-1,k-1} \end{vmatrix} \\ 
  & + \cdots + \frac{1}{2k}\begin{vmatrix} \ell_{0,0} & \ell_{0,1} & \cdots & (N+0)\ell_{0,k-1} \\ \ell_{1,0} & \ell_{1,1} & \cdots & (N-1)\ell_{1,k-1} \\ \vdots & \vdots & & \vdots \\ \ell_{2,0} & \ell_{2,1} & \cdots & (N-k+1)\ell_{k-1,k-1} \end{vmatrix}.
\end{split}
\endeq
A series of straightforward determinant manipulations now shows that 
$\mathcal{L}'(0) = \frac{N}{2}\mathcal{L}(0)$, which proves the proposition.
First, notice that using multilinearity we can write each determinant in 
\eqref{Lprime-at-0} as the sum of $N\mathcal{L}(0)/2k$ and an 
$N$-independent term.  As there are $k$ determinants, the $N$-dependent 
terms sum to $N\mathcal{L}(0)/2$, the expected answer.  We now show the 
$N$-independent terms cancel, using $k=3$ and $k=2$ to illustrate the odd 
and even cases, respectively.  

If $k$ is odd, we use multilinearity to shift all the coefficients of 
$\ell_{i,j}$ to be the same as the center summand, canceling extra terms in 
pairs (we drop the common factor of $1/2k$):
\eq
\label{k-odd-calc}
\begin{split}
& \begin{vmatrix} 2\ell_{0,0} & \ell_{0,1} & \ell_{0,2} \\ 1\ell_{1,0} & \ell_{1,1} & \ell_{1,2} \\ 0\ell_{2,0} & \ell_{2,1} & \ell_{2,2} \end{vmatrix} + \begin{vmatrix} \ell_{0,0} & 1\ell_{0,1} & \ell_{0,2} \\ \ell_{1,0} & 0\ell_{1,1} & \ell_{1,2} \\ \ell_{2,0} & -1\ell_{2,1} & \ell_{2,2} \end{vmatrix} + \begin{vmatrix} \ell_{0,0} & \ell_{0,1} & 0\ell_{0,2} \\ \ell_{1,0} & \ell_{1,1} & -1\ell_{1,2} \\ \ell_{2,0} & \ell_{2,1} & -2\ell_{2,2} \end{vmatrix} \\
= & \begin{vmatrix} (1+1)\ell_{0,0} & \ell_{0,1} & \ell_{0,2} \\ (0+1)\ell_{1,0} & \ell_{1,1} & \ell_{1,2} \\ (-1+1)\ell_{2,0} & \ell_{2,1} & \ell_{2,2} \end{vmatrix} + \begin{vmatrix} \ell_{0,0} & 1\ell_{0,1} & \ell_{0,2} \\ \ell_{1,0} & 0\ell_{1,1} & \ell_{1,2} \\ \ell_{2,0} & -1\ell_{2,1} & \ell_{2,2} \end{vmatrix} + \begin{vmatrix} \ell_{0,0} & \ell_{0,1} & (1-1)\ell_{0,2} \\ \ell_{1,0} & \ell_{1,1} & (0-1)\ell_{1,2} \\ \ell_{2,0} & \ell_{2,1} & (-1-1)\ell_{2,2} \end{vmatrix} \\
= & \begin{vmatrix} 1\ell_{0,0} & \ell_{0,1} & \ell_{0,2} \\ 0\ell_{1,0} & \ell_{1,1} & \ell_{1,2} \\ -1\ell_{2,0} & \ell_{2,1} & \ell_{2,2} \end{vmatrix} + \begin{vmatrix} \ell_{0,0} & 1\ell_{0,1} & \ell_{0,2} \\ \ell_{1,0} & 0\ell_{1,1} & \ell_{1,2} \\ \ell_{2,0} & -1\ell_{2,1} & \ell_{2,2} \end{vmatrix} + \begin{vmatrix} \ell_{0,0} & \ell_{0,1} & 1\ell_{0,2} \\ \ell_{1,0} & \ell_{1,1} & 0\ell_{1,2} \\ \ell_{2,0} & \ell_{2,1} & -1\ell_{2,2} \end{vmatrix}.
\end{split}
\endeq
Each row now has an associated coefficient $j$, 
$j\in\{-(k-1)/2,...,(k-1)/2\}$.  These coefficients sum to zero, so we can 
rewrite $\mathcal{L}(0)$ by multiplying each row by $x^j$:
\eq
\begin{vmatrix} \ell_{0,0} & \ell_{0,1} & \ell_{0,2} \\ \ell_{1,0} & \ell_{1,1} & \ell_{1,2} \\ \ell_{2,0} & \ell_{2,1} & \ell_{2,2} \end{vmatrix} = \begin{vmatrix} x^1\ell_{0,0} & x^1\ell_{0,1} & x^1\ell_{0,2} \\ x^0\ell_{1,0} & x^0\ell_{1,1} & x^0\ell_{1,2} \\ x^{-1}\ell_{2,0} & x^{-1}\ell_{2,1} & x^{-1}\ell_{2,2} \end{vmatrix}.
\endeq
Taking the derivative of both sides with respect to $x$ and then setting 
$x=1$ shows the $N$-independent terms (those in \eqref{k-odd-calc}) are zero.

If $k$ is even, then there is no center summand.  We now shift the 
coefficients so the coefficient of the top row is $k/2$ and that of the 
bottom row is $-k/2+1$.  Note that one summand already has these 
coefficients, so while most shifts will cancel in pairs, the final summand 
introduces a new term $-(k/2)\mathcal{L}(0)$:
\eq
\begin{vmatrix} 1\ell_{0,0} & \ell_{0,1} \\ 0\ell_{1,0} & \ell_{1,1} \end{vmatrix} + \begin{vmatrix} \ell_{0,0} & 0\ell_{0,1} \\ \ell_{1,0} & -1\ell_{1,1} \end{vmatrix} = \begin{vmatrix} 1\ell_{0,0} & \ell_{0,1} \\ 0\ell_{1,0} & \ell_{1,1} \end{vmatrix} + \begin{vmatrix} \ell_{0,0} & 1\ell_{0,1} \\ \ell_{1,0} & 0\ell_{1,1} \end{vmatrix} - 1\begin{vmatrix} \ell_{0,0} & \ell_{0,1} \\ \ell_{1,0} & \ell_{1,1} \end{vmatrix}.
\endeq
We can now rewrite $x^{k/2}\mathcal{L}(0)$ as $\mathcal{L}(0)$ with each 
row multiplied by $x^j$, where $j$ is the appropriate coefficient:
\eq
x^1\begin{vmatrix} \ell_{0,0} & \ell_{0,1} \\ \ell_{1,0} & \ell_{1,1} \end{vmatrix} = \begin{vmatrix} x^1\ell_{0,0} & x^1\ell_{0,1} \\ x^0\ell_{1,0} & x^0\ell_{1,1} \end{vmatrix}.
\endeq
Differentiating both sides with respect to $x$ and setting $x=1$ shows 
the $N$-independent terms are zero if $k$ is even.
\section{The Barnes \texorpdfstring{$G$}{TEXT}-function\label{Barnes}}
 The Barnes $G$-function is defined as the infinite product
\beq
 G\left(1+z\right) := \left(2\pi\right)^{\frac{z}{2}}\exp\left(-\frac{z+z^2\left(1+\gamma\right)}{2}\right)\prod_{k=1}^{\infty}
 \left(1+\frac{z}{k}\right)^k \exp\left(\frac{z^2}{2k}-z\right),
\eeq
 where $\gamma$ is Euler's constant.
  The Barnes $G$-function satisfies the functional equation 
\beq
G\left(1+z\right)=\Gamma\left(z\right)G\left(z\right).
\eeq
It  is analytic in the whole complex plane and has
   the following asymptotic expansion as $|z|\rightarrow\infty$, $\mathrm{arg}\,z\neq\pi$:
  \beq
  \log G(1+z)=\left(\frac{z^2}{2}-\frac{1}{12}\right)\log z-\frac{3z^2}{4}
  +\frac{z}{2}\log 2\pi+\zeta'(-1)+\mathcal{O}\left(z^{-2}\right).
  \eeq
   It satisfies the useful relation
 \beq\label{barnesaux1}
\dfrac{ G(1+z+n)G(1-z)}{G( 1-z-n)G(1+z)}=\left(-1\right)^{\frac{n(n+1)}{2}}\left(\frac{\pi}{\sin\pi z}\right)^n,\qquad n\in\mathbb{Z}.
 \eeq


\begin{thebibliography}{99}

\bibitem{BBBCPRS}
E. Bailey, S. Bettin, G. Blower, J.B. Conrey, A, Prokhorov, M.O. Rubinstein and N.C. Snaith, 
Mixed moments of characteristic polynomials of random unitary matrices, arXiv:1901.07479.

\bibitem{BleherIts:1999}
P. Bleher and A.R. Its,
Semiclassical asymptotics of orthogonal polynomials, Riemann-Hilbert problem, and universality in the matrix model,
{\it Ann. Math.}
{\bf 150},
185--266
(1999).

\bibitem{BornemannFM:2017}
F. Bornemann, P.J. Forrester, and A. Mays, 
Finite size effects for spacing distributions in random matrix theory: circular ensembles and Riemann zeros,
{\it Stud. Appl. Math.} 
{\bf 138}, 
401--437
(2017).

\bibitem{BumpG}
D. Bump and A. Gamburd,
On the averages of characteristic polynomials from the classical groups,
{\it Comm. Math. Phys.} 
{\bf 265}, 
227--274 
(2006).

\bibitem{Clarkson2005}
 P. Clarkson, 
Special polynomials associated with rational solutions of the fifth Painlev\'e  equation, 
{\it J. Comput. Appl. Math.} 
{\bf  178}, 
111--129 
(2005).

\bibitem{Clarkson2013}
P. Clarkson,  
Recurrence coefficients for discrete orthonormal polynomials and the Painlev\'e equations,
{\it J. Phys. A}
{\bf  46},  
185205 
(2013).

\bibitem{ConreyFKRS:2003}
J.B. Conrey, D.W. Farmer, J.P. Keating, M.O. Rubinstein, and N.C. Snaith, 
Autocorrelation of random matrix polynomials,
{\it Comm. Math. Phys.} 
{\bf 237}, 
365--395 
(2003).

\bibitem{ConreyFKRS:2005}
J.B. Conrey, D.W. Farmer, J.P. Keating, M.O. Rubinstein, and N.C. Snaith, 
Integral moments of L-functions,
{\it Proc. Lond. Math. Soc. (3)} 
{\bf 91}, 
33--104 
(2005).

\bibitem{ConreyFKRS:2008}
J.B. Conrey, D.W. Farmer, J.P. Keating, M.O. Rubinstein, and N.C. Snaith, 
Lower order terms in the full moment conjecture for the Riemann zeta function,
{\it J. Number Theory} 
{\bf 128},
1516--1554 
(2008).

\bibitem{ConreyRS:2006}
J.B. Conrey, M.O. Rubinstein, and N.C. Snaith,
Moments of the derivative of characteristic polynomials with an application to the Riemann zeta function,
{\it Comm. Math. Phys.}
{\bf 267},
611--629
(2006).

\bibitem{Conway:1990}
J. Conway, 
{\it A Course in Functional Analysis},
Springer-Verlag, 
New York, 
1990.

\bibitem{Dehaye:2008} 
P.-O. Dehaye, 
Joint moments of derivatives  of characteristic polynomials,
{\it Algebra Number Theory} 
{\bf 2}, 
31--68 
(2008).

\bibitem{Dehaye:2010}
P.-O. Dehaye, 
A note on moments of derivatives of characteristic polynomials,
22nd International Conference on Formal Power Series and Algebraic Combinatorics (FPSAC 2010), 681--692, 
{\it Discrete Math. Theor. Comput. Sci. Proc., AN}, 
Assoc. Discrete Math. Theor. Comput. Sci., Nancy, 2010. 

\bibitem{DeiftZ:1995}
P.A. Deift and X. Zhou, 
Asymptotics for the Painlev\'e II equation,
{\it Comm. Pure Appl. Math.} 
{\bf 48},
277--337
(1995).

\bibitem{Dyson:1962}
F.J. Dyson, 
Statistical theory of the energy levels of complex systems I, II, III, 
{\it J. Math. Phys.}
{\bf 3}, 
140--175
(1962).

\bibitem{FokasIKN:2006}
A.S. Fokas, A.R. Its, A.A. Kapaev, and V.Y. Novokshenov, 
{\it Painlev\'e Transcendents. The Riemann-Hilbert Approach.} 
Mathematical Surveys and Monographs 
{\bf 128}. 
American Mathematical Society, 
Providence, RI, 
2006.

\bibitem{ForresterW:2002}
P.J. Forrester and N.S. Witte,
Application of the $\tau$-function theory of Painlev\'e equations to random matrices:  $P_V$, $P_{III}$, the LUE, JUE, and CUE,
{\it Comm. Pure Appl. Math.}
{\bf 55},
679--727
(2002).

\bibitem{ForresterW:2006}
P.J. Forrester and N.S. Witte,
Boundary conditions associated with the Painlev\'e III$^\prime$ and V evaluations of some random matrix averages,
{\it J. Phys. A}
{\bf 39},
8983--8995
(2006).

\bibitem{ForresterW:2015}
P.J. Forrester and N.S. Witte,
Painlev\'e II in random matrix theory and related fields,
{\it Constr. Approx.} 
{\bf 41}, 
589--613 
(2015).

\bibitem{Lisovyy1}
O. Gamayun, N. Iorgov, and O. Lisovyy,
How instanton combinatorics solves Painlev\'e~VI,~V~and~IIIs, 
{\it J. Phys. A} 
{\bf 46}, 
335203
(2013).

\bibitem{Hughes1}
C.P. Hughes,
On the characteristic polynomial of a random unitary matrix and the Riemann zeta function,
Ph.D. Thesis, University of Bristol
(2001).

\bibitem{Hughes2}
C.P. Hughes,
Random matrix theory and discrete moments of the Riemann zeta function,
{\it J. Phys. A}
{\bf 36},
2907
(2003).

\bibitem{HughesKO:2000}
C.P. Hughes, J.P. Keating, and N. O'Connell, 
Random matrix theory and the derivative of the Riemann zeta function, 
{\it Proc. R. Soc. Lond. A} 
{\bf 456}, 
2611--2627 
(2000).

\bibitem{HughesKO:2001}
C.P. Hughes, J.P. Keating, and N. O'Connell,  
On the characteristic polynomial of a random unitary matrix, 
{\it Comm. Math. Phys.} 
{\bf 220}, 
429--451 
(2001).

\bibitem{its} 
A.R. Its, 
Large $N$ asymptotics in random matrices, 
in 
{\it Random Matrices, Random Processes and Integrable Systems}, 
CRM Series in Mathematical Physics, 
ed. John Harnad, 
Springer, 
2011.

\bibitem{Its:2011}
A.R. Its,
Painlev\'e transcendents, 
in 
{\it The Oxford Handbook of Random Matrix Theory},
eds. G. Akemann, J. Baik, and P. Di Francesco,
Oxford University Press,
2011.

\bibitem{JM} 
M. Jimbo and T. Miwa, 
Monodromy preserving deformations of linear ordinary differential equations with rational coefficients. II, 
{\it Physica D} 
{\bf 2},
407--448
(1981).

\bibitem{JMU} 
M. Jimbo, T. Miwa, and K. Ueno, 
Monodromy preserving deformation of linear ordinary differential equations with rational coefficients. I. General theory and $\tau$-function,
{\it Physica D} 
{\bf 2}, 
306--352
(1981).

\bibitem{KM}
K. Kajiwara and T. Masuda, 
On the Umemura polynomials for the Painlev\'e III equation, 
{\it Phys. Lett. A}
 {\bf 260}, 
462--467
(1999).

\bibitem{Keating2005}
J.P. Keating, 
Random matrices and number theory, 
in {\it Proceedings of the Les Houches Summer School on Applications of Random Matrices in Physics, Les Houches 2004}, eds. E. Brezin et al., Springer, 
2005.

\bibitem{Keating2017}
J.P. Keating,
Random matrices and number theory: some recent themes, 
in {\it Stochastic Processes and Random Matrices: Lecture Notes of the Les Houches Summer School: Volume 104, July 2015}, 
eds. G. Schehr, A. Altland, Y.V. Fyodorov, N. O'Connell, and L.F. Cugliandolo,
2017.

\bibitem{KeatingS:2000}
J.P. Keating and N.C. Snaith,
Random matrix theory and $\zeta(1/2+it)$,
{\it Comm. Math. Phys.}
{\bf 214},
57--85
(2000).

\bibitem{KeatingS:2000b}
J.P. Keating and N.C. Snaith,
Random matrix theory and $L$-functions at $s=1/2$,
{\it Comm. Math. Phys.}
{\bf 214},
91--110
(2000).

\bibitem{KeatingS:2011}
J.P. Keating and N.C. Snaith, 
Random matrix theory and number theory, 
in {\it The Handbook on Random Matrix Theory}, eds. G. Akemann, J. Baik, and P. Di Francesco, 
Oxford University Press,
2011.

\bibitem{Lisovyy}
O. Lisovyy, H. Nagoya, and J. Roussillon, 
 Irregular conformal blocks and connection formulae for Painlev\'e V functions, 
{\it J. Math Phys.}
{\bf 59},
091409
(2018).

\bibitem{Masuda:2004}
T. Masuda,
Classical transcendental solutions of the Painlev\'e equations and their degeneration,
{\it Tohoku Math. J.}
{\bf 56},
467--490
(2004).

\bibitem{Masuda}
T. Masuda, Y. Ohta, and K. Kajiwara, 
A determinant formula for a class of rational solutions of Painlev\'e V equation, 
{\it Nagoya Math. J.} 
{\bf  168},  
1--25 
(2002).

\bibitem{Montgomery:1973}
H.L. Montgomery, 
The pair correlation of zeros of the zeta function, 
{\it Proc. Symp. Pure Math.} 
{\bf 24},
181--193
(1973). 

\bibitem{NoumiY:1998}
M. Noumi and Y. Yamada, 
Umemura polynomials for the Painlev\'e V equation,
{\it Phys. Lett.}
{\bf A247},
65--69
(1998).

\bibitem{Okamoto:1980a}
K. Okamoto, 
Polynomial Hamiltonians associated with Painlev\'e equations. I, 
{\it Proc. Japan Acad. Ser. A Math. Sci.} 
{\bf 56}, 
264--268
(1980).

\bibitem{Okamoto:1980b}
K. Okamoto, 
Polynomial Hamiltonians associated with Painlev\'e equations. II, 
{\it Proc. Japan Acad. Ser. A Math. Sci.} 
{\bf 56}, 
367--371
(1980).

\bibitem{Okamoto:1987}
K. Okamoto, 
Studies on the Painlev\'e equations.  II.  Fifth Painlev\'e equation $P_V$,
{\it Japan. J. Math.} 
{\bf 13}, 
47--76
(1987).

\bibitem{Okamoto:1987b}
K. Okamoto,
Studies on the Painlev\'e equations IV.  Third Painlev\'e equation $P_{III}$,
{\it Funkcial. Ekvac.}
{\bf 30},
305--332
(1987).

\bibitem{Riedtmann}
H. Riedtmann,
A combinatorial approach to mixed ratios of characteristic polynomials,
arXiv:1805.07261
(2018).

\bibitem{Snaith:2010}
N.C. Snaith,
Riemann zeros and random matrix theory,
{\it Milan J. Math.}
{\bf 78},
135--152
(2010).

\bibitem{Winn:2012} 
B. Winn, 
Derivative moments for  characteristic polynomials from the CUE,
{\it Comm. Math. Phys.} 
{\bf 315}, 
531--562 
(2012).

\end{thebibliography}
\end{document}